\documentclass[11pt]{article}
\usepackage[hidelinks]{hyperref}
\usepackage{amsmath,amssymb,amsthm,graphicx}
\usepackage{algorithm}
\usepackage{algpseudocode}

\usepackage[dvipsnames]{xcolor}
\usepackage{algpseudocode}
\usepackage{float,subfig}
\usepackage{enumitem}
\usepackage{caption}
\usepackage{tikz}
\usetikzlibrary{calc, decorations.pathmorphing}
\usepackage{pgfplots}
\pgfplotsset{scaled y ticks = false, compat = newest, legend={down then right}}

\makeatletter
\pgfplotsset{
  every axis x label/.append style={
    alias=current axis xlabel
  },
  legend pos/outer south/.style={
    /pgfplots/legend style={
      at={%
        (%
        \@ifundefined{pgf@sh@ns@current axis xlabel}%
        {xticklabel cs:0.5}%
        {current axis xlabel.south}%
        )%
      },
      anchor=north
    }
  }
}
\makeatother

\usepackage[textsize=small]{todonotes}

\usepackage[yyyymmdd]{datetime}

\usepackage[margin=36pt]{caption}

\newtheorem{claim}{Claim}
\newtheorem{lemma}{Lemma}

\theoremstyle{definition}

\usepackage[a4paper, total={6.2in, 9.5in}]{geometry}

\usepackage{titling}

\title{Efficient Algorithms and Implementations for Extracting Maximum-Size $(k,\ell)$-Sparse Subgraphs 
}
\date{\vspace*{-12pt}}

\author{
  P\'eter Madarasi\thanks{HUN-REN Alfr\'{e}d R\'{e}nyi Institute of Mathematics, and Department of Operations Research, E{\"o}tv{\"o}s Lor{\'a}nd University, Budapest, Hungary. E-mail: \texttt{madarasip@staff.elte.hu}}
}

\begin{document}

\maketitle

\begin{abstract}
  A multigraph $G = (V, E)$ is $(k, \ell)$-sparse if every subset $X \subseteq V$ induces at most $\max\{k|X| - \ell, 0\}$ edges.
  Finding a maximum-size $(k, \ell)$-sparse subgraph is a classical problem in rigidity theory and combinatorial optimization, with known polynomial-time algorithms.
  This paper presents a highly efficient and flexible implementation of an augmenting path method, enhanced with a range of powerful practical heuristics that significantly reduce running time while preserving optimality.
  These heuristics --- including edge-ordering, node-ordering, two-phase strategies, and pseudoforest-based initialization --- steer the algorithm toward accepting more edges early in the execution and avoiding costly augmentations.
  A comprehensive experimental evaluation on both synthetic and real-world graphs demonstrates that our implementation outperforms existing tools by several orders of magnitude.
  We also propose an asymptotically faster algorithm for extracting an inclusion-wise maximal $(k,2k)$-sparse subgraph with the sparsity condition required only for node sets of size at least three, which is particularly relevant to 3D rigidity when $k = 3$.
  We provide a carefully engineered implementation, which is publicly available online and is proposed for inclusion in the LEMON graph library.

\end{abstract}

\vspace{12pt}
\section{Introduction}

Throughout this paper, let $k$ and $\ell$ be non-negative integers with $\ell \leq 2k$, which we treat as constants.
Let $G = (V, E)$ be a multigraph with $n = |V|$ nodes and $m = |E|$ edges.
For $0 \leq \ell < 2k$, the graph $G = (V, E)$ is called \emph{$(k, \ell)$-sparse} if, for each subset $X \subseteq V$, the number of edges induced by $X$ is at most $\max\{k|X| - \ell, 0\}$.
For $\ell = 2k$, this condition is required only for node sets of size at least three.
Furthermore, if $G$ is $(k, \ell)$-sparse and has exactly $\max\{k |V| - \ell, 0\}$ edges, then we say that $G$ is \emph{$(k, \ell)$-tight}.
A graph is \emph{$(k, \ell)$-spanning} if it contains a $(k, \ell)$-tight subgraph that spans the entire node set.
A \emph{$(k, \ell)$-block} of a $(k, \ell)$-sparse graph is a subset $X \subseteq V$ that induces a $(k, \ell)$-tight subgraph.
A \emph{$(k, \ell)$-component} is an inclusion-wise maximal $(k, \ell)$-block.

The notion of $(k, \ell)$-sparsity plays a central role in rigidity theory, where it serves as a combinatorial criterion for determining whether a graph is rigid or flexible in various geometric settings.
Beyond rigidity, $(k, \ell)$-sparsity also arises in combinatorial optimization --- for example, a graph is $(k, k)$-sparse if and only if its edge set can be partitioned into $k$ edge-disjoint forests.
As such, the development of efficient algorithms for testing sparsity and extracting maximum-size $(k, \ell)$-sparse subgraphs is of critical importance in both theoretical and applied contexts.

\paragraph{Historical overview}
The concept of $(k, \ell)$-sparse graphs was first introduced in 1979 by Lorea~\cite{lorea} as part of his work on matroidal families.
Since then, sparse graphs have been the focus of extensive research, with diverse applications emerging in areas such as rigidity theory, matroid theory, network design, and combinatorial optimization.
For example, $(k, k)$-tight graphs appeared in the work of Nash-Williams~\cite{nash} and Tutte~\cite{tutte} as a characterization of graphs that can be decomposed into $k$ edge-disjoint spanning trees.
A notable result in rigidity theory is due to Laman~\cite{laman}, who showed that $(2, 3)$-tight graphs are the generic minimally rigid graphs for bar--joint frameworks in the plane, while $(2, 3)$-spanning graphs are the rigid ones.
The problem of testing rigidity in 3D frameworks, however, remains open.
For further details on rigidity theory, see~\cite{jordan2016combinatorial,schulze2017rigidity}.

\paragraph{Previous work}
The classical algorithm for finding $(k, \ell)$-sparse subgraphs is based on augmenting paths and Hakimi’s Orientation Lemma~\cite{SLHOrientationLemma}.
Using this approach, one can compute a maximum-weight $(k, \ell)$-sparse subgraph in $O(nm)$ time for all $0 \leq \ell < 2k$~\cite{bergPhD, berg2003algorithms, hendrickson}.
For the maximum-size case, this algorithm can be further optimized to run in $O(n^2 + m)$ time with only $O(n)$ space~\cite{bergPhD, berg2003algorithms, hendrickson}.
A more advanced variant achieves the same running time for the weighted case, at the cost of using $O(n^2)$ space~\cite{deak2025quadratic}.
These algorithms, commonly known as pebble game algorithms~\cite{pebble, pebbleDS}, are central to rigidity theory and underpin a broad class of combinatorial algorithms.
Additional results have been established for special cases: when $\ell < 0$, the maximum-size $(k, \ell)$-sparse subgraph problem remains solvable in polynomial time~\cite{negativel}; for recognizing $(2, 3)$-sparse graphs, an $O(n \sqrt{n \log n})$-time algorithm is known~\cite{daescu2009towards}; and for recognizing planar $(2, 3)$-tight graphs, there exists an $O(n \log^3 n)$-time algorithm~\cite{rollin2019nlog3n}.
Despite these advances, no algorithms faster than the augmenting path method are known for the full range $0 \leq \ell < 2k$.
As for practical implementations, KINARI-Web~\cite{kinariArticle, kinari} --- developed by Fox, Jagodzinski, Li, and Streinu --- is, to the best of our knowledge, the only efficient (closed-source \texttt{C++}) tool for computing both maximum-size and maximum-weight $(k, \ell)$-sparse subgraphs.
Another closely related \texttt{C++} implementation by Cs.\ Kir\'aly and Mih\'alyk\'o~\cite{rigidityAugm} is capable of making a $(k, \ell)$-tight (hyper)graph $(k, \ell)$-redundant and can also find a maximum-size $(k, \ell)$-sparse subgraph.

\paragraph{Our contribution}
We present an efficient and flexible \texttt{C++} implementation of several variants of the augmenting path algorithm for computing maximum-size $(k, \ell)$-sparse subgraphs.
Our implementation supports both weighted and unweighted input graphs and includes optimized routines for extracting sparse subgraphs, checking sparsity, and identifying $(k,\ell)$-components.
We conduct an extensive computational study comparing our implementation for solving the maximum-size $(k, \ell)$-sparse subgraph problem to two existing tools: the KINARI-Web library~\cite{kinari} and a tool by Cs.\ Kir\'aly and Mih\'alyk\'o~\cite{rigidityAugm}.
The experiments span a wide variety of synthetic and real-world graphs, including random, molecular, and protein graphs.
The results show that our implementation consistently outperforms existing solutions by several orders of magnitude.
To improve practical performance further in the unweighted case, we introduce a collection of edge-ordering heuristics that fine-tune algorithmic behavior in different input scenarios, as first proposed in~\cite{madarasi2023klSparse}.
These heuristics significantly reduce the running time in practice while preserving the optimality guarantees of the underlying algorithm.
Our implementation is publicly available online~\cite{githubSparse} and is proposed for inclusion in the open-source LEMON library~\cite{lemonCikk, lemon}, providing a fast, scalable, and open-source tool for rigidity theory and related applications.
Additionally, we propose an asymptotically faster algorithm that computes an inclusion-wise maximal $(k,2k)$-sparse subgraph, where sparsity is only required for node sets of size at least three.
This case is especially relevant in 3D rigidity theory, particularly for analyzing block--hole graphs with a single block~\cite{jordan2023RigidBlockAndHoleGraphs}, and yields a necessary condition for 3D rigidity when $k = 3$.

\bigskip
The rest of this paper is organized as follows.
Section~\ref{sec:preliminaries} introduces the core algorithmic framework, namely the augmenting path method, which serves as the foundation for our subsequent developments.
We define the scope of the problem, describe the main algorithmic ideas, and discuss complexity bounds and optimization opportunities.
In Section~\ref{sec:l2k}, we propose a new algorithm for the special case $\ell = 2k$, where sparsity is required only for node sets of size at least three.
Section~\ref{sec:implementation} describes our implementation in detail.
We present our practical improvements over the basic augmenting path algorithm, including an early-termination condition and other performance-critical enhancements.
In Section~\ref{sec:heuristics}, we propose a range of heuristics for the processing order of the edges, which significantly influence the practical efficiency of the algorithms.
Section~\ref{sec:benchmark} presents a comprehensive empirical evaluation of our implementation.
We benchmark the performance of different heuristics on various families of random graphs and real-world molecular and protein graphs.
We also compare our results against two existing tools.
Finally, Section~\ref{sec:conclusion} summarizes our contributions and outlines directions for future research.

\section{The augmenting path algorithm}\label{sec:preliminaries}

In this paper, we focus on the following three problems:
\begin{enumerate}
\item \emph{Decision.} Decide whether a graph is $(k, \ell)$-sparse, $(k, \ell)$-tight, $(k, \ell)$-spanning, or none of these.
\item \emph{Extraction.} Extract a maximum-size $(k, \ell)$-sparse subgraph from a given graph.
\item \emph{Components.} Find all $(k, \ell)$-components of a $(k, \ell)$-sparse graph.
\end{enumerate}

We primarily restrict ourselves to the natural range $0 \leq \ell < 2k$, but we also consider the problem of extracting an inclusion-wise maximal $(k, \ell)$-sparse subgraph when $\ell = 2k$ (with $k > 0$).

\medskip
The augmenting path algorithm for finding a maximum-size $(k, \ell)$-sparse subgraph processes the edges of the input graph one by one, accepting or rejecting each edge.
During the execution of the algorithm, an inner digraph $D$ is maintained, which is initialized as an empty graph on the same node set as $G$.
An edge $uv$ is accepted for inclusion in the sparse subgraph if and only if adding $uv$ to the undirected version of $D$ keeps the graph sparse.
To verify this, the algorithm attempts to reorient $D$ such that the sum of the indegrees of $u$ and $v$ is strictly smaller than $(2k - \ell)$, while maintaining that the indegrees of all other nodes in $D$ are at most $k$.
By the Orientation Lemma~\cite{SLHOrientationLemma}, such an orientation exists if and only if the edge $uv$ can be accepted.

If the number of arcs entering $u$ and $v$ is at least $(2k - \ell)$, then the algorithm attempts to reverse a path starting from a node with indegree smaller than $k$ and ending at either $u$ or $v$, thereby decreasing the number of incoming arcs to $u$ or $v$.
This process terminates after at most $\ell + 1$ path reversals, either when a proper orientation is found (and $uv$ is accepted) or when no further reversible paths exist (and $uv$ is rejected).
Whenever an edge $uv$ is accepted, it is inserted into $D$ and oriented such that the indegrees of $u$ and $v$ remain at most $k$.
It is easy to see that the algorithm terminates after $O(nm)$ steps.
For a more detailed description of this algorithm, the reader is referred to~\cite{deak2025quadratic}.

For $0 \leq \ell < 2k$, the collection of $(k, \ell)$-sparse edge sets forms the family of independent sets of a matroid~\cite{lorea}.
Consequently, the algorithm identifies a maximum-size $(k, \ell)$-sparse subgraph regardless of the order in which the edges are considered.
Moreover, when the edges are processed in non-increasing order of weight, the same greedy approach yields a maximum-weight $(k, \ell)$-sparse subgraph.
In contrast, when $\ell = 2k$, the $(k, \ell)$-sparse edge sets do not form a matroid, and the greedy strategy is no longer guaranteed to be optimal.
In this case, the algorithm produces only an inclusion-wise maximal $(k, \ell)$-sparse subgraph, which may be suboptimal in terms of both size and weight.

\medskip
For $0 \leq \ell < 2k$, the algorithm can be sped up by maintaining the $(k, \ell)$-components formed by the accepted edges.
It is easy to prove that any two $(k, \ell)$-components are disjoint when $0 \leq \ell \leq k$, and may intersect in at most one node when $k < \ell < 2k$.
Using this observation, the $(k, \ell)$-components can be represented efficiently, allowing the algorithm to reject an edge in constant time by checking if it is already induced by a $(k, \ell)$-component in the underlying undirected graph of the inner digraph.
Accepting an edge takes linear time in the number of nodes.
This optimization improves the running time of the algorithm to $O(n^2 + m)$.
For a detailed description of the optimized algorithm, the reader is referred to~\cite{deak2025quadratic}.
It is important to note that the component-based algorithm requires the edges to be processed in a specific order when $k < \ell$.
In this case, the edges are processed by iterating over the node set and considering all unprocessed edges incident to the current node.
This special processing order complicates the extension of this idea to the weighted case when $k < \ell$, since the weight function dictates the processing order of the edges.
An $O(n^2 + m)$-time algorithm that supports arbitrary edge order in the weighted case is described in~\cite{deak2025quadratic}.

\section{An improved algorithm for the case $\ell = 2k$}\label{sec:l2k}

In this section, we first discuss an $O(n^2 m)$-time algorithm for extracting an inclusion-wise maximal $(k, 2k)$-sparse subgraph, followed by an improved algorithm that runs in $O(n m)$ time. 
Recall that for $\ell = 2k$, we only require the sparsity condition for node subsets of size at least three, since if the condition were imposed on all subsets, then the only $(k, 2k)$-sparse graphs would be those without any edges.
This is in line with applications, particularly in 3D rigidity theory, where for $\ell = 2k$ the sparsity condition is imposed only on subsets of size at least three.
In contrast to the general multigraph setting of the previous sections, here we assume that $G$ is simple.
We begin with the following lemma.
\begin{lemma}\label{lem:l2klemma}
  Let $G=(V,E)$ be a $(k,2k)$-sparse simple graph and let $u,v\in V$ be two of its nodes.
  Assume that there is no edge between $u$ and $v$.
  Let $D$ be an orientation of $G$ in which all indegrees are at most $k$, and the indegrees of $u$ and $v$ are zero.
  Then, $G + uv$ is $(k,2k)$-sparse if and only if, for each $w \in V\setminus\{u,v\}$, there exists a directed path to $w$ from some node distinct from $u$ and $v$ with indegree smaller than $k$.
\end{lemma}
\begin{proof}
  Since $G$ is $(k,2k)$-sparse, $G$ can always be oriented in such a way that the indegree of each node is at most $k$, and the indegrees of $u$ and $v$ are zero by the Orientation Lemma~\cite{SLHOrientationLemma}.

  We now prove the lemma.
  First, suppose that there exists a path to each node in $V\setminus\{u,v\}$ from a node distinct from $u$ and $v$ with indegree smaller than $k$.
  For a subset $X$ of the nodes containing $u,v$ and a third node $w$, take a path to $w$ from a node distinct from $u$ and $v$ with indegree smaller than $k$, and reverse it.
  Since the indegree of every node remains at most $k$, the indegrees of $u$ and $v$ are zero, and the indegree of $w$ is smaller than $k$, we get that
  \[
    i(X) \leq \sum_{v\in X}\varrho(v) < (|X|-2)k = k|X|-2k,
  \]
  where $i(X)$ is the number of edges induced by $X$, and $\varrho(v)$ is the indegree of the node $v$.
  This means that $X$ is not tight, therefore $G+uv$ is $(k, 2k)$-sparse.

  Second, suppose that for a node $w$, there exist no paths to $w$ from any nodes with indegree smaller than $k$ distinct from $u$ and $v$.
  Let $R$ denote the set of nodes from which $w$ is reachable in $D$.
  We prove that $X = R \cup \{u, v\}$ is a tight set, hence it prevents the insertion of the edge $uv$.
  Observe that the indegree of every node in $R$ is exactly $k$, while the indegrees of $u$ and $v$ are zero.
  Therefore,
  \[
    i(X) = \sum_{v\in X}\varrho(v) = (|X|-2)k = k|X|-2k,
  \]
  hence $G+uv$ is not $(k,2k)$-sparse, which was to be shown.
\end{proof}

Note that this proof yields a straightforward algorithm (also proposed by Cs.\ Király) which runs in $O(n^2m)$ time.

\medskip

Now, we present an improved version of this algorithm for extracting an inclusion-wise maximal $(k,2k)$-sparse subgraph running in $O(nm)$ time, also based on Lemma~\ref{lem:l2klemma}. 
To achieve this, instead of traversing the graph $n$ times to process an edge, we execute only one BFS from the nodes with indegree smaller than $k$ in an orientation $D$ of $G$ that we maintain during the algorithm and check whether all nodes are reached.

The detailed description of the algorithm is the following.


\begin{algorithm}[H]
  \addtolength\linewidth{-8ex}
  \caption{Extract an inclusion-wise maximal $(k,2k)$-sparse subgraph, improved version}\label{alg:l2kimproved}
  \begin{algorithmic}[1]
  \item[] \textbf{Input:} A simple graph $G=(V,E)$ on at least three nodes and a positive integer~$k$.
  \item[] \textbf{Output:} An inclusion-wise maximal $(k,2k)$-sparse subgraph of $G$.
  \medskip
  \item[] 
    Construct a new digraph $D$ on the node set $V$ without any arcs.
    Then, process each edge $uv$ in an arbitrary order as follows.

    \noindent Reorient $D$ such that the indegrees of $u$ and $v$ are zero, and the indegrees of the rest of the nodes remain at most $k$, which takes at most $2k$ path reversals.
    Run a BFS in $D$ from the nodes with indegree smaller than $k$, distinct from $u$ and $v$.

    \begin{itemize}
    \item If every node in $V\setminus\{u,v\}$ is reached, then accept the edge $uv$, and insert it into $D$ with arbitrary orientation.
    \item Otherwise, there exists a node $w$ that was not reached, and hence the edge $uv$ cannot be inserted by Lemma~\ref{lem:l2klemma}.
    \end{itemize}
    After all edges are processed, output the set of accepted edges.
  \end{algorithmic}
\end{algorithm}

\paragraph{Complexity}
One reorientation requires at most $2k$ path reversals in $O(n)$ time, and the further BFS calls take $O(n)$ time for each edge. 
Hence, the algorithm takes $O(nm)$ steps in total.

\section{Implementation details and practical enhancements}\label{sec:implementation}

In this section, we describe the algorithms we implemented and highlight practical enhancements designed to improve their performance.
The implementation of all the algorithms presented here is available online~\cite{githubSparse}.
The algorithms implemented include:
\begin{enumerate}\itemsep-2pt
\item Basic augmenting path algorithm~\cite{pebble} for the maximum-weight problem for the range $0 \leq \ell < 2k$, with several practical improvements.
\item Algorithm~\ref{alg:l2kimproved} for extracting an inclusion-wise maximal $(k, 2k)$-sparse subgraph.
\item $(k, \ell)$-component-based augmenting path algorithm for the maximum-size problem for $0 \leq \ell < 2k$ as described in~\cite{deak2025quadratic}.
\item $(k, \ell)$-component-based augmenting path algorithm for the weighted case as described in~\cite{deak2025quadratic}.
\end{enumerate}
In this paper, we focus primarily on the maximum-size $(k, \ell)$-sparse subgraph problem.

All our algorithms are designed with a user-friendly interface that allows for step-by-step execution control.
Additionally, query functions provide access to all relevant information generated during the execution of the algorithms.
For more details, please refer to the documentation of the implementation~\cite{githubSparse}.

We now highlight three key performance enhancements in our implementation of the augmenting path algorithms for the maximum-size $(k, \ell)$-sparse subgraph problem.

\paragraph{Early-termination condition}
One crucial optimization in our implementation is the early-termination condition.
Once the total number of arcs in the inner digraph $D$ reaches $(k|V| - \ell)$, no further edges can be accepted because the underlying undirected graph is $(k, \ell)$-tight.
This condition significantly reduces the running time, particularly in dense graphs, since a maximum-size sparse subgraph is often found after processing only a small portion of the edges.

\paragraph{Breadth-first search}
Repeated graph traversal, especially breadth-first search (BFS), is a critical subroutine in the augmenting path algorithm.
In the standard BFS implementation within the LEMON library, the algorithm iterates over all nodes and reinitializes their visited status during each execution.
This can lead to redundant operations, especially because the algorithm is run multiple times on the same graph and only a small subset of nodes is typically explored.
In our enhanced version, we modified the BFS to reinitialize only the nodes that were visited, thus avoiding unnecessary iterations over irrelevant nodes.

\paragraph{Processing order of the edges} The most impactful insight that enhances the performance of our implementation is the ability to control the order in which the edges are processed.
For instance, when testing $(1, 1)$-sparsity on a path, processing the edges sequentially from one endpoint to the other results in a linear-time algorithm.
In contrast, processing the edges in an alternating fashion from the center gives a quadratic-time algorithm.
This demonstrates that we can achieve asymptotic speedups by carefully selecting the order in which edges are processed.
As we will discuss in the next sections, leveraging this flexibility allows the algorithm to achieve better performance in practice.

\section{Edge-ordering heuristics}\label{sec:heuristics}

In this section, we explore heuristics for determining the processing order of the edges in the augmenting path algorithm for the maximum-size $(k, \ell)$-sparse subgraph problem.
The most time-consuming operation in the algorithm is finding reversible paths.
Therefore, a natural strategy is to prioritize edges whose insertion requires the fewest path reversals.
Minimizing the number of reversals also means that edges that are ``most likely'' to be accepted are processed first, thus increasing the chances of meeting the early-termination condition after processing only a few edges.
Furthermore, since the orientation of each accepted edge can be freely chosen, we also consider different rules for orienting accepted edges.
For each heuristic, we describe the most efficient orientation strategy to complement the edge-ordering approach.

\subsubsection*{Edge orders}
Each heuristic in this section selects the next edge to be processed, removes it from the input graph, and inserts it into the inner digraph if it is accepted.
Once accepted, the selected edge is oriented according to the specific rule associated with the heuristic.
\begin{enumerate}[leftmargin=*]
\item[] \textsc{Basic}:
  Edges are processed in a random order as specified by the graph representation.
  The accepted edges are inserted into the inner digraph with a random orientation.
  This heuristic serves as the baseline in our experiments.
\item[] \textsc{DegMin}:
  Selects an edge incident to a node with minimal degree in the input graph.
  The accepted edges are oriented towards the node of minimal degree.
\item[] \textsc{IncProcMin}:
  Selects an edge whose endpoints have minimal total number of incident processed edges in the input graph.
  The accepted edges are oriented towards the endpoint with the smaller number of incident processed edges.
\item[] \textsc{IncInDegMin}:
  Selects an edge whose endpoints have minimal total indegree in the inner digraph.
  The accepted edges are oriented towards the endpoint with the smaller indegree.
\end{enumerate}

\subsubsection*{Node orders}
In this section, we introduce heuristics that iteratively select a node and process all incident edges that have not yet been processed.
These heuristics can also be used in the $(k, \ell)$-component-based algorithm, and we denote their component-based variants by the suffix \textsc{Comp} to distinguish them from edge-based heuristics.
\begin{enumerate}[leftmargin=*]
\item[] \textsc{NBasic}:
  Iterates over the nodes in random order, as specified by the graph representation, and processes all incident edges for each node.
  The accepted edges are oriented towards the current node (and outwards from the current node in the \textsc{Comp} version).
\item[] \textsc{NDegMin}:
  Selects a node with minimal degree in the input graph and processes all its incident edges.
  The accepted edges are oriented towards the current node (also in the \textsc{Comp} version).
\item[] \textsc{NProcMin}:
  Selects a node with minimal number of incident processed edges and processes all incident edges.
  The accepted edges are oriented towards the current node (outwards from the current node in the \textsc{Comp} version).
\item[] \textsc{NInDegMin}:
  Selects a node with minimal indegree in the inner digraph and processes all incident edges.
  The accepted edges are oriented outwards from the current node (towards the current node in the \textsc{Comp} version).
\end{enumerate}

\subsubsection*{A two-phase approach}
We introduce a family of heuristics based on a two-phase process.
In the first phase, we construct a large $(k, \ell)$-sparse subgraph $H$, and in the second phase, we extend it to a maximum-size $(k, \ell)$-sparse subgraph using the augmenting path method.
This approach can potentially speed up the algorithm, provided that we can efficiently construct the subgraph $H$ and orient its edges so that the augmenting path method can handle the remaining edges.
We consider several implementations of the first phase for efficiently constructing $H$.

\paragraph{Forests}
One variant of the two-phase approach starts by constructing a collection of disjoint spanning forests and pseudoforests (i.e.\ undirected graphs in which each connected component contains at most one cycle) in the first phase.
The following claim is key to ensuring that this approach gives a correct result.
\begin{claim}
  The union of $\min\{\ell, 2k - \ell\}$ forests and $(k-\ell)^+$ pseudoforests that are pairwise edge-disjoint is $(k, \ell)$-sparse.
\end{claim}
\begin{proof}
  For $0 \leq \ell \leq k$, we have the following inequality for any non-empty subset $X$ of nodes:
  \[
    i(X) \leq \min\{\ell, 2k - \ell\} (|X| - 1) + {(k - \ell)}^+ |X| = \ell (|X| - 1) + {(k - \ell)} |X| = k |X| - \ell.
  \]
  For $k < \ell$, the inequality becomes:
  \begin{align*}
    i(X) &\leq \min\{\ell, 2k - \ell\} (|X| - 1) + {(k - \ell)}^+ |X| = (2k - \ell) (|X| - 1)\\
         &= 2k|X| - \ell|X| - 2k + \ell = k|X| - \ell + (k-\ell)(|X|-2) \leq k|X| - \ell.
  \end{align*}
  This holds for each subset $X$ of size at least two.
  The condition also holds for subsets of size one, as forests are loop-free by definition.
\end{proof}

In the first phase, we orient the edges of each forest to form an inbranching, typically obtained through a graph traversal algorithm.
For pseudoforests, we adjust the orientation by reversing a path from one of the endpoints of the extra edge in the component to the root of the component.
This ensures that the extra edge can be oriented such that all indegrees remain at most one.
The digraph resulting from this construction serves as a valid inner digraph for the augmenting path method.
In the second phase, we apply the \textsc{Basic} heuristic to extract the maximum-size $(k, \ell)$-sparse subgraph.

Utilizing the freedom in choosing which forests are extended to pseudoforests, we first find the pseudoforests, then the forests, which maximizes the number of edges accepted in the first phase in practice.

We also have flexibility in how the (pseudo)forests are constructed.
Several different methods, such as BFS, DFS, and the union-find data structure, can be used to find the spanning forests.
Additionally, the order in which edges are inserted into the forests can affect the resulting structure, hence we consider various versions of the union-find approach to explore this flexibility.

We now introduce several variations of the heuristic based on different methods for constructing the (pseudo)forests.
\begin{enumerate}[leftmargin=*]
\item[] \textsc{PForestsBFS}:
  Finds $\min\{\ell, 2k - \ell\}$ spanning forests and $(k-\ell)^+$ pseudoforests using BFS.
\item[] \textsc{PForestsDFS}:
  Finds $\min\{\ell, 2k - \ell\}$ spanning forests and $(k-\ell)^+$ pseudoforests using DFS.
\item[] \textsc{ForestsBFS}:
  Finds $\min\{\ell, 2k - \ell\} + (k-\ell)^+$ spanning forests using BFS.
\item[] \textsc{ForestsDFS}:
  Finds $\min\{\ell, 2k - \ell\} + (k-\ell)^+$ spanning forests using DFS.
\item[] \textsc{UnionBasic}:
  Finds $\min\{\ell, 2k - \ell\} + (k-\ell)^+$ spanning forests, uses the \textsc{Basic} order.
\item[] \textsc{UnionNBasic}:
  Finds $\min\{\ell, 2k - \ell\} + (k-\ell)^+$ spanning forests, uses the \textsc{NBasic} order.
\item[] \textsc{UnionTranspOne}:
  Finds $\min\{\ell, 2k - \ell\} + (k-\ell)^+$ spanning forests, uses the \textsc{TranspOne} order, see the next paragraph.
\end{enumerate}

\subsubsection*{Transposed edge orders}
In the first phase of the two-phase approach, we can run the augmenting path method without performing any augmentations.
This means that edges are only accepted if the orientation of the inner digraph permits it without requiring any path reversals.
In the second phase, the remaining edges are processed in the usual manner.
Although the first phase significantly reduces the number of necessary graph traversals, it can still prove inefficient, especially because the early-termination condition is rarely met.
This often necessitates iterating over the entire edge set, which can result in quadratic time complexity in dense graphs.
To overcome this inefficiency, we introduce \emph{transposed} edge orders, which combine the two phases into one, avoiding the need to iterate through the entire edge set.
\begin{itemize}[leftmargin=*]
\item[] \textsc{Transp}:
  Cyclically iterates over the nodes and selects the next unprocessed edge incident to the current node, and then moves to the next node.
  The accepted edges are oriented towards the current node.
\item[] \textsc{TranspOne}:
  Cyclically iterates over the nodes and selects the next unprocessed edge incident to the current node as long as an edge is accepted at each node (or no more incident edges exist), before moving to the next node.
  The accepted edges are oriented towards the current node.
\end{itemize}

\medskip
The choice of edge-processing order significantly impacts the efficiency of the augmenting path algorithm.
Our heuristics aim to minimize computational overhead by prioritizing edges that are easier to process.
While these heuristics are motivated by properties of sparse graphs, their practical performance varies depending on the graph structure and density.
In the next section, we conduct an extensive empirical evaluation of these heuristics and compare them against existing solutions to assess their efficiency across various graph families.

\section{Performance evaluation of heuristics}\label{sec:benchmark}
We first describe the benchmarking environment used for testing.
The task we consider is the extraction problem, that is, we extract a maximum-size $(k, \ell)$-sparse subgraph.
The types of graphs used in the benchmarks are as follows:
\begin{enumerate}[leftmargin=*]
\item[] \textbf{Erd\H os--R\'enyi}: Erd\H os--R\'enyi random graphs, where each edge is independently included with probability $p$.
 We test the values $p \in \{0.05, 0.1, 0.5, 0.9\}$.

\item[] \textbf{Barab\'asi--Albert}: Barab\'asi--Albert scale-free graphs, generated by starting with an empty graph on $m$ nodes.
 New nodes are added one by one, each selecting $m$ distinct neighbors from existing nodes with probability proportional to their degrees.
 We test for $m \in \{5, 50\}$.

\item[] \textbf{Rigid}: Rigid graphs in the plane, that is, $(2,3)$-spanning graphs, generated using the process suggested by Cs.\ Kir\'aly and implemented by Mih\'alyk\'o.
  This involves first sampling three independent random labeled spanning trees on the same node set $V$, and letting $T=(V,E)$ be their multigraph union.
  For each node $v \in V$, we then replace $v$ by a complete graph on $d_T(v) + 1$ nodes and reattach the edges originally incident to $v$ so that each such edge is incident to a distinct node of the clique, leaving exactly one node with no reattached edges.

\item[] \textbf{Tight}: $(k, k)$-tight graphs, obtained by taking the multigraph union of $k$ random (labeled) spanning trees on the same node set.
  We consider the case $k=3$.

\item[] \textbf{Molecules and proteins}: Real-world molecular and protein graphs, extracted from the Protein Data Bank (\texttt{https://www.rcsb.org}).
  The 3D rigidity of molecular and protein graphs can be determined using the molecular conjecture~\cite{molecularConjecture}, which involves finding six disjoint spanning trees in the graph formed by taking every edge with multiplicity five.
Therefore, we modify the graphs accordingly.
\end{enumerate}

For each random graph type and a given number of nodes, we generate 10 graphs and compute the average running time for all pairs $(k, \ell)$, where $1 \leq k \leq 10$ and $0 \leq \ell < 2k$.
For molecules and proteins, we take the average running time over all instances on the given number of nodes.
In each figure below, the horizontal axis shows the number of nodes in the considered graphs and the vertical axis shows the average running time in seconds.
The legend entries are sorted by performance on the largest problem instances.
The benchmarks were executed on a machine with an AMD Ryzen 9 3950X CPU, $32$ GB RAM, running a Linux operating system.

\begin{figure}[H]
  \setcounter{subfigure}{0}
  \centering
  \subfloat[$p=0.05$] {
    \hspace{5mm}
    \pgfplotstableread[col sep=comma]{./results/heuristics/erdos_renyi/multiplicity_1/prob_0.05/Full/all.csv}{\pltA}
    \begin{tikzpicture}[scale=.95, trim axis left, trim axis right]
      \begin{axis}[
        xmin = 100.0, xmax = 1000.0,
        ymin = 0, ymax = 0.25,
        xtick distance = 200.0,
        x tick label style={/pgf/number format/1000 sep = \kern 0.15em},
        tick label style={font=\scriptsize},
        label style={font=\scriptsize},
        grid = both,
        minor tick num = 4,
        major grid style = {lightgray},
        minor grid style = {lightgray!25},
        legend style={nodes={scale=0.5, transform shape},legend pos = outer south,legend columns=3,reverse legend},
        xtick={100.0, 200.0, 300.0, 400.0, 500.0, 600.0, 700.0, 800.0, 900.0, 1000.0}
        ]
        \addplot[green, mark = oplus, mark size=1.75pt] table [x = {n}, y = {Basic1}] {\pltA};
        \addlegendentry{\textsc{Basic}}
        \addplot[orange, mark = halfsquare left*, mark size=2pt] table [x = {n}, y = {UnionTranspOneOut1}] {\pltA};
        \addlegendentry{\textsc{UnionTranspOne}}
        \addplot[red, mark = halfsquare right*, mark size=2pt] table [x = {n}, y = {UnionNodeBasic1}] {\pltA};
        \addlegendentry{\textsc{UnionNBasic}}
        \addplot[magenta, mark = halfdiamond*, mark size=2pt] table [x = {n}, y = {UnionBasic1}] {\pltA};
        \addlegendentry{\textsc{UnionBasic}}
        \addplot[brown, mark = halfcircle*, mark size=1.75pt] table [x = {n}, y = {PForestsDFS21}] {\pltA};
        \addlegendentry{\textsc{PForestsDFS}}
        \addplot[purple, mark = pentagon*, mark size=2pt] table [x = {n}, y = {ForestsDFSExplicit1}] {\pltA};
        \addlegendentry{\textsc{ForestsDFS}}
        \addplot[green, mark = triangle*, mark size=2pt] table [x = {n}, y = {ForestsBFSExplicit1}] {\pltA};
        \addlegendentry{\textsc{ForestsBFS}}
        \addplot[blue, mark = diamond*, mark size=2pt] table [x = {n}, y = {PForestsBFS21}] {\pltA};
        \addlegendentry{\textsc{PForestsBFS}}
        \addplot[green, mark = heart, mark size=2pt] table [x = {n}, y = {NodeBasicComp1}] {\pltA};
        \addlegendentry{\textsc{NBasicComp}}
        \addplot[blue, mark = otimes, mark size=1.75pt] table [x = {n}, y = {DegMin0}] {\pltA};
        \addlegendentry{\textsc{DegMin}}
        \addplot[blue, mark = Mercedes star, mark size=2pt] table [x = {n}, y = {NodeDeg0}] {\pltA};
        \addlegendentry{\textsc{NDegMin}}
        \addplot[red, mark = otimes*, mark size=1.75pt] table [x = {n}, y = {NodeInDegMin1}] {\pltA};
        \addlegendentry{\textsc{NInDegMin}}
        \addplot[purple, mark = Mercedes star flipped, mark size=2pt] table [x = {n}, y = {NodeDegComp0}] {\pltA};
        \addlegendentry{\textsc{NDegMinComp}}
        \addplot[orange, mark = halfcircle, mark size=1.75pt] table [x = {n}, y = {NodeBasic0}] {\pltA};
        \addlegendentry{\textsc{NBasic}}
        \addplot[purple, mark = square, mark size=1.5pt] table [x = {n}, y = {IncProcMin0}] {\pltA};
        \addlegendentry{\textsc{IncProcMin}}
        \addplot[magenta, mark = oplus*, mark size=1.75pt] table [x = {n}, y = {NodeProcMinComp1}] {\pltA};
        \addlegendentry{\textsc{NProcMinComp}}
        \addplot[orange, mark = square*, mark size=1.5pt] table [x = {n}, y = {NodeInDegMinComp0}] {\pltA};
        \addlegendentry{\textsc{NInDegMinComp}}
        \addplot[brown, mark = 10-pointed star, mark size=2pt] table [x = {n}, y = {NodeProcMin0}] {\pltA};
        \addlegendentry{\textsc{NProcMin}}
        \addplot[brown, mark = triangle, mark size=2pt] table [x = {n}, y = {IncInDegMin0}] {\pltA};
        \addlegendentry{\textsc{IncInDegMin}}
        \addplot[red, mark = pentagon, mark size=2pt] table [x = {n}, y = {TranspOneOut0}] {\pltA};
        \addlegendentry{\textsc{TranspOne}}
        \addplot[magenta, mark = diamond, mark size=2pt] table [x = {n}, y = {Transp0}] {\pltA};
        \addlegendentry{\textsc{Transp}}
      \end{axis}
    \end{tikzpicture}
  }%
  \hfill
  \subfloat[$p=0.1$] {
    \pgfplotstableread[col sep=comma]{./results/heuristics/erdos_renyi/multiplicity_1/prob_0.10/Full/all.csv}{\pltA}
    \begin{tikzpicture}[scale=.95, trim axis left, trim axis right]
      \begin{axis}[
        xmin = 100.0, xmax = 1000.0,
        ymin = 0, ymax = 0.25,
        xtick distance = 200.0,
        x tick label style={/pgf/number format/1000 sep = \kern 0.15em},
        tick label style={font=\scriptsize},
        label style={font=\scriptsize},
        grid = both,
        minor tick num = 4,
        major grid style = {lightgray},
        minor grid style = {lightgray!25},
        legend style={nodes={scale=0.5, transform shape},legend pos = outer south,legend columns=3,reverse legend},
        xtick={100.0, 200.0, 300.0, 400.0, 500.0, 600.0, 700.0, 800.0, 900.0, 1000.0}
        ]
        \addplot[green, mark = oplus, mark size=1.75pt] table [x = {n}, y = {Basic1}] {\pltA};
        \addlegendentry{\textsc{Basic}}
        \addplot[orange, mark = halfsquare left*, mark size=2pt] table [x = {n}, y = {UnionTranspOneOut1}] {\pltA};
        \addlegendentry{\textsc{UnionTranspOne}}
        \addplot[red, mark = halfsquare right*, mark size=2pt] table [x = {n}, y = {UnionNodeBasic1}] {\pltA};
        \addlegendentry{\textsc{UnionNBasic}}
        \addplot[magenta, mark = halfdiamond*, mark size=2pt] table [x = {n}, y = {UnionBasic1}] {\pltA};
        \addlegendentry{\textsc{UnionBasic}}
        \addplot[purple, mark = pentagon*, mark size=2pt] table [x = {n}, y = {ForestsDFSExplicit1}] {\pltA};
        \addlegendentry{\textsc{ForestsDFS}}
        \addplot[brown, mark = halfcircle*, mark size=1.75pt] table [x = {n}, y = {PForestsDFS21}] {\pltA};
        \addlegendentry{\textsc{PForestsDFS}}
        \addplot[green, mark = triangle*, mark size=2pt] table [x = {n}, y = {ForestsBFSExplicit1}] {\pltA};
        \addlegendentry{\textsc{ForestsBFS}}
        \addplot[blue, mark = diamond*, mark size=2pt] table [x = {n}, y = {PForestsBFS21}] {\pltA};
        \addlegendentry{\textsc{PForestsBFS}}
        \addplot[green, mark = heart, mark size=2pt] table [x = {n}, y = {NodeBasicComp1}] {\pltA};
        \addlegendentry{\textsc{NBasicComp}}
        \addplot[purple, mark = square, mark size=1.5pt] table [x = {n}, y = {IncProcMin0}] {\pltA};
        \addlegendentry{\textsc{IncProcMin}}
        \addplot[blue, mark = otimes, mark size=1.75pt] table [x = {n}, y = {DegMin0}] {\pltA};
        \addlegendentry{\textsc{DegMin}}
        \addplot[brown, mark = triangle, mark size=2pt] table [x = {n}, y = {IncInDegMin0}] {\pltA};
        \addlegendentry{\textsc{IncInDegMin}}
        \addplot[purple, mark = Mercedes star flipped, mark size=2pt] table [x = {n}, y = {NodeDegComp0}] {\pltA};
        \addlegendentry{\textsc{NDegMinComp}}
        \addplot[magenta, mark = oplus*, mark size=1.75pt] table [x = {n}, y = {NodeProcMinComp1}] {\pltA};
        \addlegendentry{\textsc{NProcMinComp}}
        \addplot[orange, mark = square*, mark size=1.5pt] table [x = {n}, y = {NodeInDegMinComp0}] {\pltA};
        \addlegendentry{\textsc{NInDegMinComp}}
        \addplot[blue, mark = Mercedes star, mark size=2pt] table [x = {n}, y = {NodeDeg0}] {\pltA};
        \addlegendentry{\textsc{NDegMin}}
        \addplot[red, mark = otimes*, mark size=1.75pt] table [x = {n}, y = {NodeInDegMin1}] {\pltA};
        \addlegendentry{\textsc{NInDegMin}}
        \addplot[orange, mark = halfcircle, mark size=1.75pt] table [x = {n}, y = {NodeBasic0}] {\pltA};
        \addlegendentry{\textsc{NBasic}}
        \addplot[brown, mark = 10-pointed star, mark size=2pt] table [x = {n}, y = {NodeProcMin0}] {\pltA};
        \addlegendentry{\textsc{NProcMin}}
        \addplot[red, mark = pentagon, mark size=2pt] table [x = {n}, y = {TranspOneOut0}] {\pltA};
        \addlegendentry{\textsc{TranspOne}}
        \addplot[magenta, mark = diamond, mark size=2pt] table [x = {n}, y = {Transp0}] {\pltA};
        \addlegendentry{\textsc{Transp}}
      \end{axis}
    \end{tikzpicture}
    \hspace{5mm}
  }
  \newline
  \centering
  \subfloat[$p=0.5$] {
    \hspace{5mm}
    \pgfplotstableread[col sep=comma]{./results/heuristics/erdos_renyi/multiplicity_1/prob_0.50/Full/all.csv}{\pltA}
    \begin{tikzpicture}[scale=.95, trim axis left, trim axis right]
      \begin{axis}[
        xmin = 100.0, xmax = 1000.0,
        ymin = 0, ymax = 0.5,
        xtick distance = 200.0,
        x tick label style={/pgf/number format/1000 sep = \kern 0.15em},
        tick label style={font=\scriptsize},
        label style={font=\scriptsize},
        grid = both,
        minor tick num = 4,
        major grid style = {lightgray},
        minor grid style = {lightgray!25},
        legend style={nodes={scale=0.5, transform shape},legend pos = outer south,legend columns=3,reverse legend},
        xtick={100.0, 200.0, 300.0, 400.0, 500.0, 600.0, 700.0, 800.0, 900.0, 1000.0}
        ]
        \addplot[green, mark = oplus, mark size=1.75pt] table [x = {n}, y = {Basic1}] {\pltA};
        \addlegendentry{\textsc{Basic}}
        \addplot[orange, mark = halfsquare left*, mark size=2pt] table [x = {n}, y = {UnionTranspOneOut1}] {\pltA};
        \addlegendentry{\textsc{UnionTranspOne}}
        \addplot[red, mark = halfsquare right*, mark size=2pt] table [x = {n}, y = {UnionNodeBasic1}] {\pltA};
        \addlegendentry{\textsc{UnionNBasic}}
        \addplot[magenta, mark = halfdiamond*, mark size=2pt] table [x = {n}, y = {UnionBasic1}] {\pltA};
        \addlegendentry{\textsc{UnionBasic}}
        \addplot[purple, mark = pentagon*, mark size=2pt] table [x = {n}, y = {ForestsDFSExplicit1}] {\pltA};
        \addlegendentry{\textsc{ForestsDFS}}
        \addplot[brown, mark = halfcircle*, mark size=1.75pt] table [x = {n}, y = {PForestsDFS21}] {\pltA};
        \addlegendentry{\textsc{PForestsDFS}}
        \addplot[green, mark = triangle*, mark size=2pt] table [x = {n}, y = {ForestsBFSExplicit1}] {\pltA};
        \addlegendentry{\textsc{ForestsBFS}}
        \addplot[blue, mark = diamond*, mark size=2pt] table [x = {n}, y = {PForestsBFS21}] {\pltA};
        \addlegendentry{\textsc{PForestsBFS}}
        \addplot[purple, mark = square, mark size=1.5pt] table [x = {n}, y = {IncProcMin0}] {\pltA};
        \addlegendentry{\textsc{IncProcMin}}
        \addplot[brown, mark = triangle, mark size=2pt] table [x = {n}, y = {IncInDegMin0}] {\pltA};
        \addlegendentry{\textsc{IncInDegMin}}
        \addplot[green, mark = heart, mark size=2pt] table [x = {n}, y = {NodeBasicComp1}] {\pltA};
        \addlegendentry{\textsc{NBasicComp}}
        \addplot[blue, mark = otimes, mark size=1.75pt] table [x = {n}, y = {DegMin0}] {\pltA};
        \addlegendentry{\textsc{DegMin}}
        \addplot[purple, mark = Mercedes star flipped, mark size=2pt] table [x = {n}, y = {NodeDegComp0}] {\pltA};
        \addlegendentry{\textsc{NDegMinComp}}
        \addplot[magenta, mark = oplus*, mark size=1.75pt] table [x = {n}, y = {NodeProcMinComp1}] {\pltA};
        \addlegendentry{\textsc{NProcMinComp}}
        \addplot[orange, mark = square*, mark size=1.5pt] table [x = {n}, y = {NodeInDegMinComp0}] {\pltA};
        \addlegendentry{\textsc{NInDegMinComp}}
        \addplot[blue, mark = Mercedes star, mark size=2pt] table [x = {n}, y = {NodeDeg0}] {\pltA};
        \addlegendentry{\textsc{NDegMin}}
        \addplot[brown, mark = 10-pointed star, mark size=2pt] table [x = {n}, y = {NodeProcMin0}] {\pltA};
        \addlegendentry{\textsc{NProcMin}}
        \addplot[red, mark = otimes*, mark size=1.75pt] table [x = {n}, y = {NodeInDegMin1}] {\pltA};
        \addlegendentry{\textsc{NInDegMin}}
        \addplot[red, mark = pentagon, mark size=2pt] table [x = {n}, y = {TranspOneOut0}] {\pltA};
        \addlegendentry{\textsc{TranspOne}}
        \addplot[orange, mark = halfcircle, mark size=1.75pt] table [x = {n}, y = {NodeBasic0}] {\pltA};
        \addlegendentry{\textsc{NBasic}}
        \addplot[magenta, mark = diamond, mark size=2pt] table [x = {n}, y = {Transp0}] {\pltA};
        \addlegendentry{\textsc{Transp}}
      \end{axis}
    \end{tikzpicture}
  }%
  \hfill
  \subfloat[$p=0.9$] {
    \pgfplotstableread[col sep=comma]{./results/heuristics/erdos_renyi/multiplicity_1/prob_0.90/Full/all.csv}{\pltA}
    \begin{tikzpicture}[scale=.95, trim axis left, trim axis right]
      \begin{axis}[
        xmin = 100.0, xmax = 1000.0,
        ymin = 0, ymax = 0.5,
        xtick distance = 200.0,
        x tick label style={/pgf/number format/1000 sep = \kern 0.15em},
        tick label style={font=\scriptsize},
        label style={font=\scriptsize},
        grid = both,
        minor tick num = 4,
        major grid style = {lightgray},
        minor grid style = {lightgray!25},
        legend style={nodes={scale=0.5, transform shape},legend pos = outer south,legend columns=3,reverse legend},
        xtick={100.0, 200.0, 300.0, 400.0, 500.0, 600.0, 700.0, 800.0, 900.0, 1000.0}
        ]
        \addplot[green, mark = oplus, mark size=1.75pt] table [x = {n}, y = {Basic1}] {\pltA};
        \addlegendentry{\textsc{Basic}}
        \addplot[orange, mark = halfsquare left*, mark size=2pt] table [x = {n}, y = {UnionTranspOneOut1}] {\pltA};
        \addlegendentry{\textsc{UnionTranspOne}}
        \addplot[red, mark = halfsquare right*, mark size=2pt] table [x = {n}, y = {UnionNodeBasic1}] {\pltA};
        \addlegendentry{\textsc{UnionNBasic}}
        \addplot[magenta, mark = halfdiamond*, mark size=2pt] table [x = {n}, y = {UnionBasic1}] {\pltA};
        \addlegendentry{\textsc{UnionBasic}}
        \addplot[purple, mark = pentagon*, mark size=2pt] table [x = {n}, y = {ForestsDFSExplicit1}] {\pltA};
        \addlegendentry{\textsc{ForestsDFS}}
        \addplot[brown, mark = halfcircle*, mark size=1.75pt] table [x = {n}, y = {PForestsDFS21}] {\pltA};
        \addlegendentry{\textsc{PForestsDFS}}
        \addplot[green, mark = triangle*, mark size=2pt] table [x = {n}, y = {ForestsBFSExplicit1}] {\pltA};
        \addlegendentry{\textsc{ForestsBFS}}
        \addplot[purple, mark = square, mark size=1.5pt] table [x = {n}, y = {IncProcMin0}] {\pltA};
        \addlegendentry{\textsc{IncProcMin}}
        \addplot[blue, mark = diamond*, mark size=2pt] table [x = {n}, y = {PForestsBFS21}] {\pltA};
        \addlegendentry{\textsc{PForestsBFS}}
        \addplot[brown, mark = triangle, mark size=2pt] table [x = {n}, y = {IncInDegMin0}] {\pltA};
        \addlegendentry{\textsc{IncInDegMin}}
        \addplot[green, mark = heart, mark size=2pt] table [x = {n}, y = {NodeBasicComp1}] {\pltA};
        \addlegendentry{\textsc{NBasicComp}}
        \addplot[blue, mark = otimes, mark size=1.75pt] table [x = {n}, y = {DegMin0}] {\pltA};
        \addlegendentry{\textsc{DegMin}}
        \addplot[purple, mark = Mercedes star flipped, mark size=2pt] table [x = {n}, y = {NodeDegComp0}] {\pltA};
        \addlegendentry{\textsc{NDegMinComp}}
        \addplot[blue, mark = Mercedes star, mark size=2pt] table [x = {n}, y = {NodeDeg0}] {\pltA};
        \addlegendentry{\textsc{NDegMin}}
        \addplot[magenta, mark = oplus*, mark size=1.75pt] table [x = {n}, y = {NodeProcMinComp1}] {\pltA};
        \addlegendentry{\textsc{NProcMinComp}}
        \addplot[orange, mark = square*, mark size=1.5pt] table [x = {n}, y = {NodeInDegMinComp0}] {\pltA};
        \addlegendentry{\textsc{NInDegMinComp}}
        \addplot[magenta, mark = diamond, mark size=2pt] table [x = {n}, y = {Transp0}] {\pltA};
        \addlegendentry{\textsc{Transp}}
        \addplot[red, mark = pentagon, mark size=2pt] table [x = {n}, y = {TranspOneOut0}] {\pltA};
        \addlegendentry{\textsc{TranspOne}}
        \addplot[red, mark = otimes*, mark size=1.75pt] table [x = {n}, y = {NodeInDegMin1}] {\pltA};
        \addlegendentry{\textsc{NInDegMin}}
        \addplot[brown, mark = 10-pointed star, mark size=2pt] table [x = {n}, y = {NodeProcMin0}] {\pltA};
        \addlegendentry{\textsc{NProcMin}}
        \addplot[orange, mark = halfcircle, mark size=1.75pt] table [x = {n}, y = {NodeBasic0}] {\pltA};
        \addlegendentry{\textsc{NBasic}}
      \end{axis}
    \end{tikzpicture}
    \hspace{5mm}
  }
  \caption{Erd\H os--R\'enyi random graphs.
    The \textsc{Transp} and \textsc{TranspOne} heuristics consistently outperform others for $p = 0.05$ and $p = 0.1$.
    For denser graphs, \textsc{NBasic}, \textsc{NProcMin}, and \textsc{NInDegMin} show competitive performance.}
\end{figure}
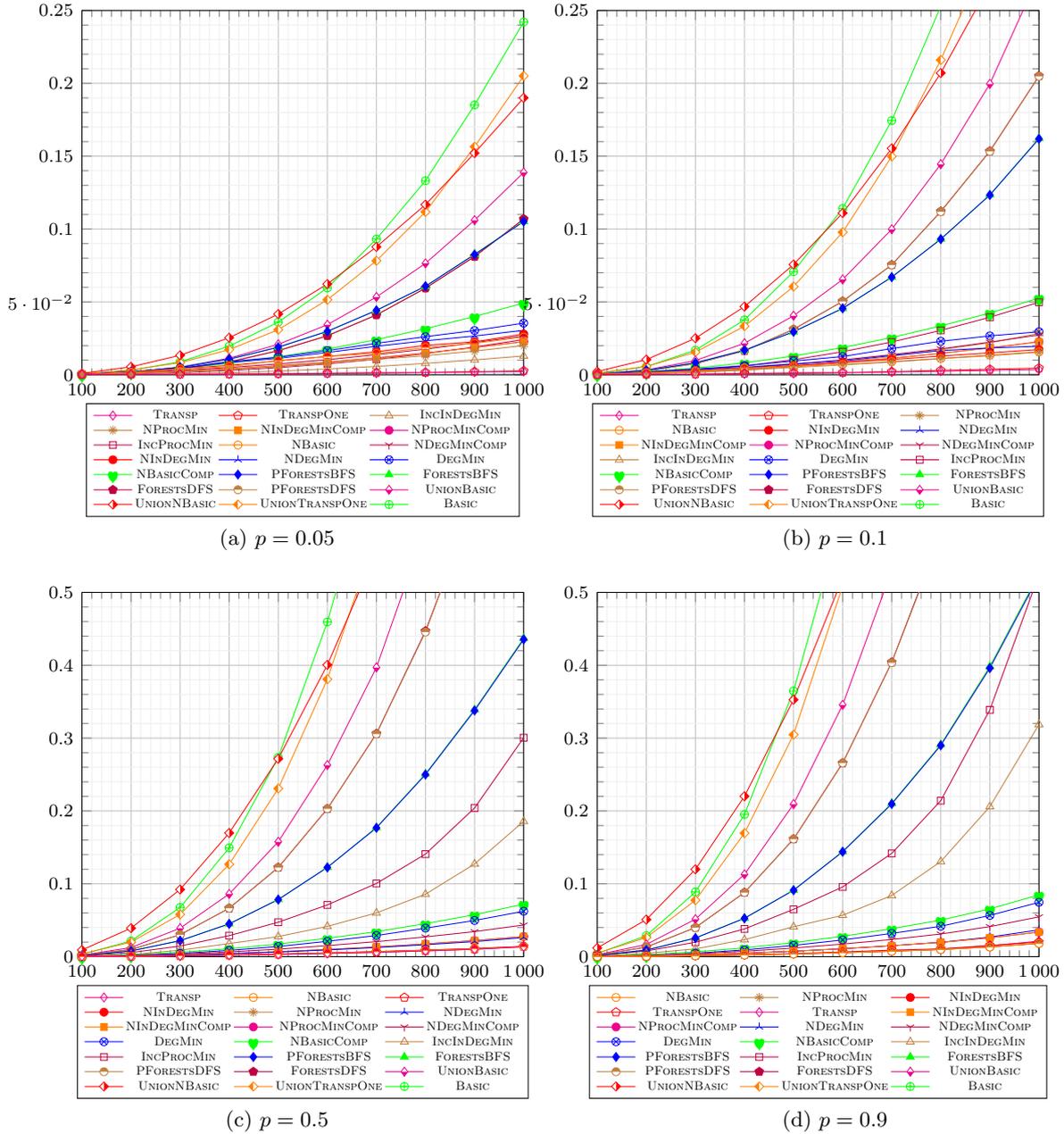

\begin{figure}[H]
  \setcounter{subfigure}{0}
  \centering
  \subfloat[$m=5$] {
    \hspace{5mm}
    \pgfplotstableread[col sep=comma]{./results/heuristics/barabasi_albert/edges_5/Full/all.csv}{\pltA}
    \begin{tikzpicture}[scale=.95, trim axis left, trim axis right]
      \begin{axis}[
        xmin = 100.0, xmax = 1000.0,
        ymin = 0, ymax = 0.025,
        xtick distance = 200.0,
        x tick label style={/pgf/number format/1000 sep = \kern 0.15em},
        tick label style={font=\scriptsize},
        label style={font=\scriptsize},
        grid = both,
        minor tick num = 4,
        major grid style = {lightgray},
        minor grid style = {lightgray!25},
        legend style={nodes={scale=0.5, transform shape},legend pos = outer south,legend columns=3,reverse legend},
        xtick={100.0, 200.0, 300.0, 400.0, 500.0, 600.0, 700.0, 800.0, 900.0, 1000.0}
        ]
        \addplot[purple, mark = Mercedes star flipped, mark size=2pt] table [x = {n}, y = {NodeDegComp0}] {\pltA};
        \addlegendentry{\textsc{NDegMinComp}}
        \addplot[green, mark = heart, mark size=2pt] table [x = {n}, y = {NodeBasicComp1}] {\pltA};
        \addlegendentry{\textsc{NBasicComp}}
        \addplot[magenta, mark = oplus*, mark size=1.75pt] table [x = {n}, y = {NodeProcMinComp1}] {\pltA};
        \addlegendentry{\textsc{NProcMinComp}}
        \addplot[orange, mark = square*, mark size=1.5pt] table [x = {n}, y = {NodeInDegMinComp0}] {\pltA};
        \addlegendentry{\textsc{NInDegMinComp}}
        \addplot[red, mark = halfsquare right*, mark size=2pt] table [x = {n}, y = {UnionNodeBasic1}] {\pltA};
        \addlegendentry{\textsc{UnionNBasic}}
        \addplot[green, mark = oplus, mark size=1.75pt] table [x = {n}, y = {Basic1}] {\pltA};
        \addlegendentry{\textsc{Basic}}
        \addplot[orange, mark = halfcircle, mark size=1.75pt] table [x = {n}, y = {NodeBasic0}] {\pltA};
        \addlegendentry{\textsc{NBasic}}
        \addplot[blue, mark = otimes, mark size=1.75pt] table [x = {n}, y = {DegMin0}] {\pltA};
        \addlegendentry{\textsc{DegMin}}
        \addplot[blue, mark = Mercedes star, mark size=2pt] table [x = {n}, y = {NodeDeg0}] {\pltA};
        \addlegendentry{\textsc{NDegMin}}
        \addplot[orange, mark = halfsquare left*, mark size=2pt] table [x = {n}, y = {UnionTranspOneOut1}] {\pltA};
        \addlegendentry{\textsc{UnionTranspOne}}
        \addplot[magenta, mark = halfdiamond*, mark size=2pt] table [x = {n}, y = {UnionBasic1}] {\pltA};
        \addlegendentry{\textsc{UnionBasic}}
        \addplot[brown, mark = halfcircle*, mark size=1.75pt] table [x = {n}, y = {PForestsDFS21}] {\pltA};
        \addlegendentry{\textsc{PForestsDFS}}
        \addplot[purple, mark = pentagon*, mark size=2pt] table [x = {n}, y = {ForestsDFSExplicit1}] {\pltA};
        \addlegendentry{\textsc{ForestsDFS}}
        \addplot[green, mark = triangle*, mark size=2pt] table [x = {n}, y = {ForestsBFSExplicit1}] {\pltA};
        \addlegendentry{\textsc{ForestsBFS}}
        \addplot[blue, mark = diamond*, mark size=2pt] table [x = {n}, y = {PForestsBFS21}] {\pltA};
        \addlegendentry{\textsc{PForestsBFS}}
        \addplot[red, mark = otimes*, mark size=1.75pt] table [x = {n}, y = {NodeInDegMin1}] {\pltA};
        \addlegendentry{\textsc{NInDegMin}}
        \addplot[purple, mark = square, mark size=1.5pt] table [x = {n}, y = {IncProcMin0}] {\pltA};
        \addlegendentry{\textsc{IncProcMin}}
        \addplot[brown, mark = 10-pointed star, mark size=2pt] table [x = {n}, y = {NodeProcMin0}] {\pltA};
        \addlegendentry{\textsc{NProcMin}}
        \addplot[brown, mark = triangle, mark size=2pt] table [x = {n}, y = {IncInDegMin0}] {\pltA};
        \addlegendentry{\textsc{IncInDegMin}}
        \addplot[magenta, mark = diamond, mark size=2pt] table [x = {n}, y = {Transp0}] {\pltA};
        \addlegendentry{\textsc{Transp}}
        \addplot[red, mark = pentagon, mark size=2pt] table [x = {n}, y = {TranspOneOut0}] {\pltA};
        \addlegendentry{\textsc{TranspOne}}
      \end{axis}
    \end{tikzpicture}
  }%
  \hfill
  \subfloat[$m=50$] {
    \pgfplotstableread[col sep=comma]{./results/heuristics/barabasi_albert/edges_50/Full/all.csv}{\pltA}
    \begin{tikzpicture}[scale=.95, trim axis left, trim axis right]
      \begin{axis}[
        xmin = 100.0, xmax = 1000.0,
        ymin = 0, ymax = 0.1,
        xtick distance = 200.0,
        x tick label style={/pgf/number format/1000 sep = \kern 0.15em},
        tick label style={font=\scriptsize},
        label style={font=\scriptsize},
        grid = both,
        minor tick num = 4,
        major grid style = {lightgray},
        minor grid style = {lightgray!25},
        legend style={nodes={scale=0.5, transform shape},legend pos = outer south,legend columns=3,reverse legend},
        xtick={100.0, 200.0, 300.0, 400.0, 500.0, 600.0, 700.0, 800.0, 900.0, 1000.0}
        ]
        \addplot[green, mark = oplus, mark size=1.75pt] table [x = {n}, y = {Basic1}] {\pltA};
        \addlegendentry{\textsc{Basic}}
        \addplot[red, mark = halfsquare right*, mark size=2pt] table [x = {n}, y = {UnionNodeBasic1}] {\pltA};
        \addlegendentry{\textsc{UnionNBasic}}
        \addplot[orange, mark = halfsquare left*, mark size=2pt] table [x = {n}, y = {UnionTranspOneOut1}] {\pltA};
        \addlegendentry{\textsc{UnionTranspOne}}
        \addplot[magenta, mark = halfdiamond*, mark size=2pt] table [x = {n}, y = {UnionBasic1}] {\pltA};
        \addlegendentry{\textsc{UnionBasic}}
        \addplot[brown, mark = halfcircle*, mark size=1.75pt] table [x = {n}, y = {PForestsDFS21}] {\pltA};
        \addlegendentry{\textsc{PForestsDFS}}
        \addplot[purple, mark = pentagon*, mark size=2pt] table [x = {n}, y = {ForestsDFSExplicit1}] {\pltA};
        \addlegendentry{\textsc{ForestsDFS}}
        \addplot[green, mark = triangle*, mark size=2pt] table [x = {n}, y = {ForestsBFSExplicit1}] {\pltA};
        \addlegendentry{\textsc{ForestsBFS}}
        \addplot[blue, mark = diamond*, mark size=2pt] table [x = {n}, y = {PForestsBFS21}] {\pltA};
        \addlegendentry{\textsc{PForestsBFS}}
        \addplot[blue, mark = otimes, mark size=1.75pt] table [x = {n}, y = {DegMin0}] {\pltA};
        \addlegendentry{\textsc{DegMin}}
        \addplot[blue, mark = Mercedes star, mark size=2pt] table [x = {n}, y = {NodeDeg0}] {\pltA};
        \addlegendentry{\textsc{NDegMin}}
        \addplot[orange, mark = halfcircle, mark size=1.75pt] table [x = {n}, y = {NodeBasic0}] {\pltA};
        \addlegendentry{\textsc{NBasic}}
        \addplot[green, mark = heart, mark size=2pt] table [x = {n}, y = {NodeBasicComp1}] {\pltA};
        \addlegendentry{\textsc{NBasicComp}}
        \addplot[purple, mark = Mercedes star flipped, mark size=2pt] table [x = {n}, y = {NodeDegComp0}] {\pltA};
        \addlegendentry{\textsc{NDegMinComp}}
        \addplot[red, mark = otimes*, mark size=1.75pt] table [x = {n}, y = {NodeInDegMin1}] {\pltA};
        \addlegendentry{\textsc{NInDegMin}}
        \addplot[magenta, mark = oplus*, mark size=1.75pt] table [x = {n}, y = {NodeProcMinComp1}] {\pltA};
        \addlegendentry{\textsc{NProcMinComp}}
        \addplot[orange, mark = square*, mark size=1.5pt] table [x = {n}, y = {NodeInDegMinComp0}] {\pltA};
        \addlegendentry{\textsc{NInDegMinComp}}
        \addplot[purple, mark = square, mark size=1.5pt] table [x = {n}, y = {IncProcMin0}] {\pltA};
        \addlegendentry{\textsc{IncProcMin}}
        \addplot[brown, mark = 10-pointed star, mark size=2pt] table [x = {n}, y = {NodeProcMin0}] {\pltA};
        \addlegendentry{\textsc{NProcMin}}
        \addplot[brown, mark = triangle, mark size=2pt] table [x = {n}, y = {IncInDegMin0}] {\pltA};
        \addlegendentry{\textsc{IncInDegMin}}
        \addplot[red, mark = pentagon, mark size=2pt] table [x = {n}, y = {TranspOneOut0}] {\pltA};
        \addlegendentry{\textsc{TranspOne}}
        \addplot[magenta, mark = diamond, mark size=2pt] table [x = {n}, y = {Transp0}] {\pltA};
        \addlegendentry{\textsc{Transp}}
      \end{axis}
    \end{tikzpicture}
    \hspace{5mm}
  }
  \caption{Barab\'asi--Albert random graphs.
    As in Erd\H os--R\'enyi graphs, \textsc{Transp} and \textsc{TranspOne} continue to be the most efficient.}
\end{figure}
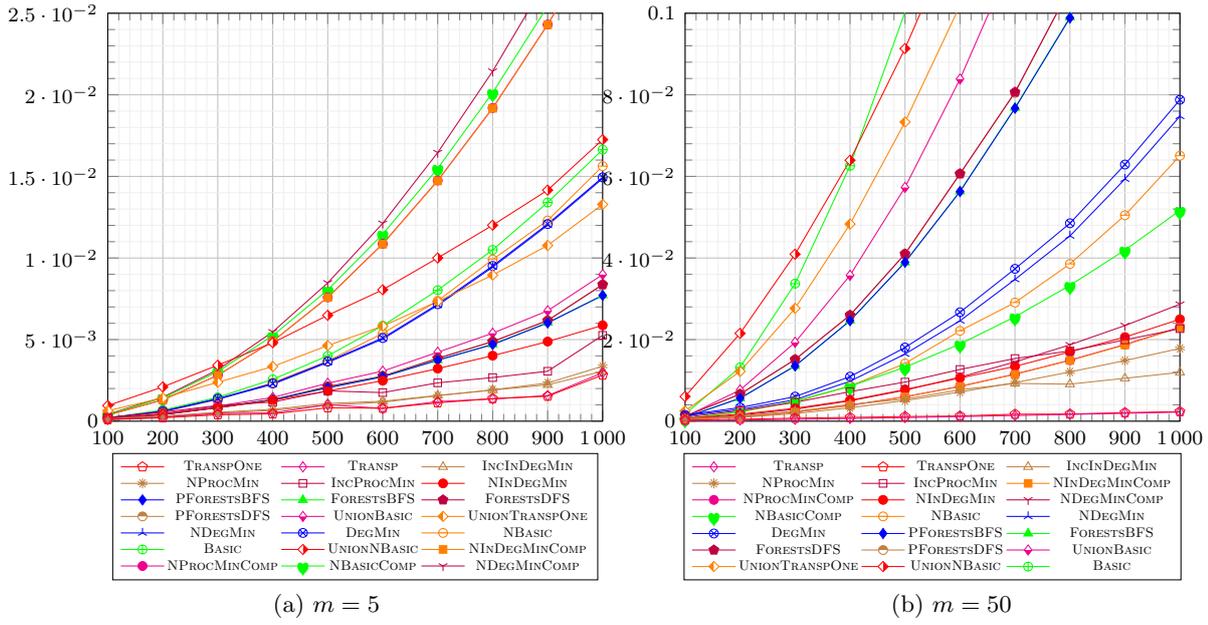

\begin{figure}[H]
  \setcounter{subfigure}{0}
  \centering
  \subfloat[Rigid graphs] {
    \hspace{5mm}
    \pgfplotstableread[col sep=comma]{./results/heuristics/rigids/Full/all.csv}{\pltA}
    \begin{tikzpicture}[scale=.95, trim axis left, trim axis right]
      \begin{axis}[
        xmin = 99.0, xmax = 1002.0,
        ymin = 0, ymax = 0.005,
        xtick distance = 200.0,
        x tick label style={/pgf/number format/1000 sep = \kern 0.15em},
        tick label style={font=\scriptsize},
        label style={font=\scriptsize},
        grid = both,
        minor tick num = 4,
        major grid style = {lightgray},
        minor grid style = {lightgray!25},
        legend style={nodes={scale=0.5, transform shape},legend pos = outer south,legend columns=3,reverse legend},
        xtick={100.0, 200.0, 300.0, 400.0, 500.0, 600.0, 700.0, 800.0, 900.0, 1000.0}
        ]
        \addplot[magenta, mark = oplus*, mark size=1.75pt] table [x = {n}, y = {NodeProcMinComp1}] {\pltA};
        \addlegendentry{\textsc{NProcMinComp}}
        \addplot[orange, mark = square*, mark size=1.5pt] table [x = {n}, y = {NodeInDegMinComp0}] {\pltA};
        \addlegendentry{\textsc{NInDegMinComp}}
        \addplot[purple, mark = Mercedes star flipped, mark size=2pt] table [x = {n}, y = {NodeDegComp0}] {\pltA};
        \addlegendentry{\textsc{NDegMinComp}}
        \addplot[green, mark = heart, mark size=2pt] table [x = {n}, y = {NodeBasicComp1}] {\pltA};
        \addlegendentry{\textsc{NBasicComp}}
        \addplot[red, mark = halfsquare right*, mark size=2pt] table [x = {n}, y = {UnionNodeBasic1}] {\pltA};
        \addlegendentry{\textsc{UnionNBasic}}
        \addplot[orange, mark = halfsquare left*, mark size=2pt] table [x = {n}, y = {UnionTranspOneOut1}] {\pltA};
        \addlegendentry{\textsc{UnionTranspOne}}
        \addplot[purple, mark = square, mark size=1.5pt] table [x = {n}, y = {IncProcMin0}] {\pltA};
        \addlegendentry{\textsc{IncProcMin}}
        \addplot[magenta, mark = halfdiamond*, mark size=2pt] table [x = {n}, y = {UnionBasic1}] {\pltA};
        \addlegendentry{\textsc{UnionBasic}}
        \addplot[blue, mark = otimes, mark size=1.75pt] table [x = {n}, y = {DegMin0}] {\pltA};
        \addlegendentry{\textsc{DegMin}}
        \addplot[blue, mark = Mercedes star, mark size=2pt] table [x = {n}, y = {NodeDeg0}] {\pltA};
        \addlegendentry{\textsc{NDegMin}}
        \addplot[brown, mark = triangle, mark size=2pt] table [x = {n}, y = {IncInDegMin0}] {\pltA};
        \addlegendentry{\textsc{IncInDegMin}}
        \addplot[green, mark = triangle*, mark size=2pt] table [x = {n}, y = {ForestsBFSExplicit1}] {\pltA};
        \addlegendentry{\textsc{ForestsBFS}}
        \addplot[blue, mark = diamond*, mark size=2pt] table [x = {n}, y = {PForestsBFS21}] {\pltA};
        \addlegendentry{\textsc{PForestsBFS}}
        \addplot[magenta, mark = diamond, mark size=2pt] table [x = {n}, y = {Transp0}] {\pltA};
        \addlegendentry{\textsc{Transp}}
        \addplot[brown, mark = halfcircle*, mark size=1.75pt] table [x = {n}, y = {PForestsDFS21}] {\pltA};
        \addlegendentry{\textsc{PForestsDFS}}
        \addplot[purple, mark = pentagon*, mark size=2pt] table [x = {n}, y = {ForestsDFSExplicit1}] {\pltA};
        \addlegendentry{\textsc{ForestsDFS}}
        \addplot[red, mark = pentagon, mark size=2pt] table [x = {n}, y = {TranspOneOut0}] {\pltA};
        \addlegendentry{\textsc{TranspOne}}
        \addplot[orange, mark = halfcircle, mark size=1.75pt] table [x = {n}, y = {NodeBasic0}] {\pltA};
        \addlegendentry{\textsc{NBasic}}
        \addplot[red, mark = otimes*, mark size=1.75pt] table [x = {n}, y = {NodeInDegMin1}] {\pltA};
        \addlegendentry{\textsc{NInDegMin}}
        \addplot[brown, mark = 10-pointed star, mark size=2pt] table [x = {n}, y = {NodeProcMin0}] {\pltA};
        \addlegendentry{\textsc{NProcMin}}
        \addplot[green, mark = oplus, mark size=1.75pt] table [x = {n}, y = {Basic1}] {\pltA};
        \addlegendentry{\textsc{Basic}}
      \end{axis}
    \end{tikzpicture}
  }%
  \hfill
  \subfloat[Tight graphs] {
    \pgfplotstableread[col sep=comma]{./results/heuristics/tights/Full/all.csv}{\pltA}
    \begin{tikzpicture}[scale=.95, trim axis left, trim axis right]
      \begin{axis}[
        xmin = 100.0, xmax = 1000.0,
        ymin = 0, ymax = 0.005,
        xtick distance = 200.0,
        x tick label style={/pgf/number format/1000 sep = \kern 0.15em},
        tick label style={font=\scriptsize},
        label style={font=\scriptsize},
        grid = both,
        minor tick num = 4,
        major grid style = {lightgray},
        minor grid style = {lightgray!25},
        legend style={nodes={scale=0.5, transform shape},legend pos = outer south,legend columns=3,reverse legend},
        xtick={100.0, 200.0, 300.0, 400.0, 500.0, 600.0, 700.0, 800.0, 900.0, 1000.0}
        ]
        \addplot[green, mark = heart, mark size=2pt] table [x = {n}, y = {NodeBasicComp1}] {\pltA};
        \addlegendentry{\textsc{NBasicComp}}
        \addplot[purple, mark = Mercedes star flipped, mark size=2pt] table [x = {n}, y = {NodeDegComp0}] {\pltA};
        \addlegendentry{\textsc{NDegMinComp}}
        \addplot[orange, mark = square*, mark size=1.5pt] table [x = {n}, y = {NodeInDegMinComp0}] {\pltA};
        \addlegendentry{\textsc{NInDegMinComp}}
        \addplot[magenta, mark = oplus*, mark size=1.75pt] table [x = {n}, y = {NodeProcMinComp1}] {\pltA};
        \addlegendentry{\textsc{NProcMinComp}}
        \addplot[red, mark = halfsquare right*, mark size=2pt] table [x = {n}, y = {UnionNodeBasic1}] {\pltA};
        \addlegendentry{\textsc{UnionNBasic}}
        \addplot[orange, mark = halfsquare left*, mark size=2pt] table [x = {n}, y = {UnionTranspOneOut1}] {\pltA};
        \addlegendentry{\textsc{UnionTranspOne}}
        \addplot[blue, mark = otimes, mark size=1.75pt] table [x = {n}, y = {DegMin0}] {\pltA};
        \addlegendentry{\textsc{DegMin}}
        \addplot[blue, mark = Mercedes star, mark size=2pt] table [x = {n}, y = {NodeDeg0}] {\pltA};
        \addlegendentry{\textsc{NDegMin}}
        \addplot[red, mark = otimes*, mark size=1.75pt] table [x = {n}, y = {NodeInDegMin1}] {\pltA};
        \addlegendentry{\textsc{NInDegMin}}
        \addplot[orange, mark = halfcircle, mark size=1.75pt] table [x = {n}, y = {NodeBasic0}] {\pltA};
        \addlegendentry{\textsc{NBasic}}
        \addplot[green, mark = oplus, mark size=1.75pt] table [x = {n}, y = {Basic1}] {\pltA};
        \addlegendentry{\textsc{Basic}}
        \addplot[magenta, mark = halfdiamond*, mark size=2pt] table [x = {n}, y = {UnionBasic1}] {\pltA};
        \addlegendentry{\textsc{UnionBasic}}
        \addplot[green, mark = triangle*, mark size=2pt] table [x = {n}, y = {ForestsBFSExplicit1}] {\pltA};
        \addlegendentry{\textsc{ForestsBFS}}
        \addplot[blue, mark = diamond*, mark size=2pt] table [x = {n}, y = {PForestsBFS21}] {\pltA};
        \addlegendentry{\textsc{PForestsBFS}}
        \addplot[purple, mark = square, mark size=1.5pt] table [x = {n}, y = {IncProcMin0}] {\pltA};
        \addlegendentry{\textsc{IncProcMin}}
        \addplot[brown, mark = halfcircle*, mark size=1.75pt] table [x = {n}, y = {PForestsDFS21}] {\pltA};
        \addlegendentry{\textsc{PForestsDFS}}
        \addplot[brown, mark = 10-pointed star, mark size=2pt] table [x = {n}, y = {NodeProcMin0}] {\pltA};
        \addlegendentry{\textsc{NProcMin}}
        \addplot[purple, mark = pentagon*, mark size=2pt] table [x = {n}, y = {ForestsDFSExplicit1}] {\pltA};
        \addlegendentry{\textsc{ForestsDFS}}
        \addplot[brown, mark = triangle, mark size=2pt] table [x = {n}, y = {IncInDegMin0}] {\pltA};
        \addlegendentry{\textsc{IncInDegMin}}
        \addplot[magenta, mark = diamond, mark size=2pt] table [x = {n}, y = {Transp0}] {\pltA};
        \addlegendentry{\textsc{Transp}}
        \addplot[red, mark = pentagon, mark size=2pt] table [x = {n}, y = {TranspOneOut0}] {\pltA};
        \addlegendentry{\textsc{TranspOne}}
      \end{axis}
    \end{tikzpicture}
    \hspace{5mm}
  }
  \caption{Rigid and tight graphs.
    For rigid graphs, the \textsc{Basic} heuristic performs best, suggesting that additional preprocessing and edge-ordering heuristics are unnecessary.
    In contrast, for $(k, k)$-tight graphs, \textsc{Transp} and \textsc{TranspOne} show superior efficiency.}
\end{figure}
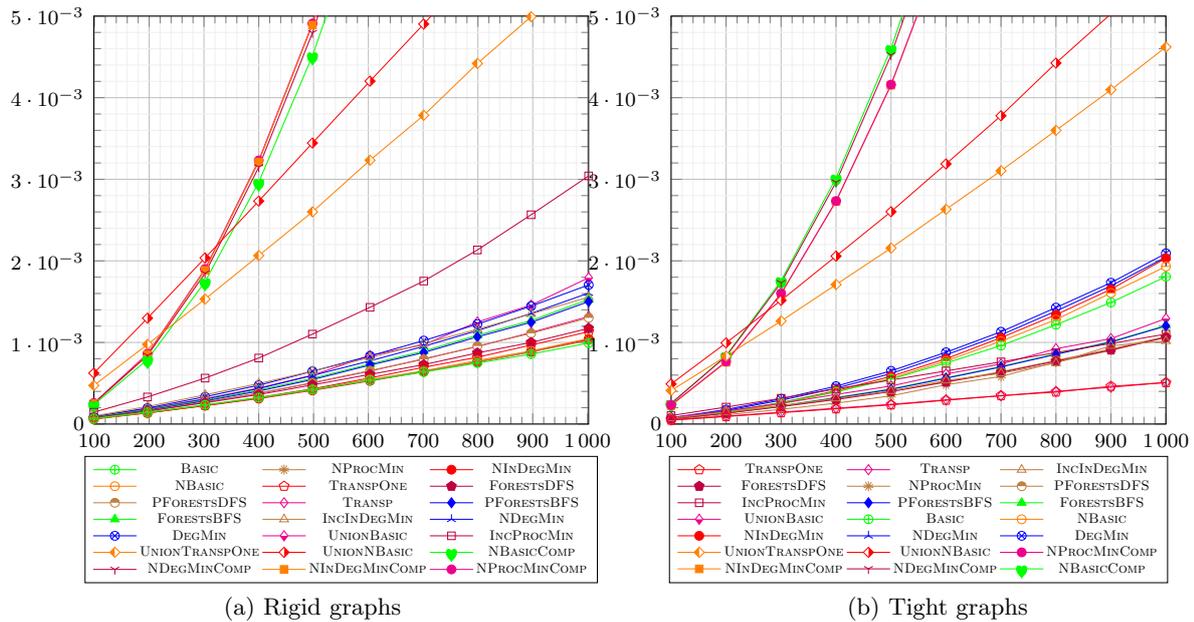

\begin{figure}[H]
  \setcounter{subfigure}{0}
  \centering
  \subfloat[Molecular graphs] {
    \hspace{5mm}
    \pgfplotstableread[col sep=comma]{./results/heuristics/moleculars/Full/all.csv}{\pltA}
    \begin{tikzpicture}[scale=.95, trim axis left, trim axis right]
      \begin{axis}[
        xmin = 8.0, xmax = 99.0,
        ymin = 0, ymax = 0.0003,
        xtick distance = 10.0,
        x tick label style={/pgf/number format/1000 sep = \kern 0.15em},
        tick label style={font=\scriptsize},
        label style={font=\scriptsize},
        grid = both,
        minor tick num = 5,
        major grid style = {lightgray},
        minor grid style = {lightgray!25},
        legend style={nodes={scale=0.5, transform shape},legend pos = outer south,legend columns=3,reverse legend},
        xtick={8.0, 20.0, 30.0, 40.0, 50.0, 60.0, 70.0, 80.0, 90.0, 99.0},
        minor xtick={8, 10, ..., 98}
        ]
        \addplot[red, mark = halfsquare right*, mark size=2pt] table [x = {n}, y = {UnionNodeBasic1}] {\pltA};
        \addlegendentry{\textsc{UnionNBasic}}
        \addplot[orange, mark = halfsquare left*, mark size=2pt] table [x = {n}, y = {UnionTranspOneOut1}] {\pltA};
        \addlegendentry{\textsc{UnionTranspOne}}
        \addplot[purple, mark = square, mark size=1.5pt] table [x = {n}, y = {IncProcMin0}] {\pltA};
        \addlegendentry{\textsc{IncProcMin}}
        \addplot[purple, mark = Mercedes star flipped, mark size=2pt] table [x = {n}, y = {NodeDegComp0}] {\pltA};
        \addlegendentry{\textsc{NDegMinComp}}
        \addplot[green, mark = heart, mark size=2pt] table [x = {n}, y = {NodeBasicComp1}] {\pltA};
        \addlegendentry{\textsc{NBasicComp}}
        \addplot[magenta, mark = oplus*, mark size=1.75pt] table [x = {n}, y = {NodeProcMinComp1}] {\pltA};
        \addlegendentry{\textsc{NProcMinComp}}
        \addplot[orange, mark = square*, mark size=1.5pt] table [x = {n}, y = {NodeInDegMinComp0}] {\pltA};
        \addlegendentry{\textsc{NInDegMinComp}}
        \addplot[brown, mark = triangle, mark size=2pt] table [x = {n}, y = {IncInDegMin0}] {\pltA};
        \addlegendentry{\textsc{IncInDegMin}}
        \addplot[magenta, mark = diamond, mark size=2pt] table [x = {n}, y = {Transp0}] {\pltA};
        \addlegendentry{\textsc{Transp}}
        \addplot[blue, mark = otimes, mark size=1.75pt] table [x = {n}, y = {DegMin0}] {\pltA};
        \addlegendentry{\textsc{DegMin}}
        \addplot[red, mark = pentagon, mark size=2pt] table [x = {n}, y = {TranspOneOut0}] {\pltA};
        \addlegendentry{\textsc{TranspOne}}
        \addplot[magenta, mark = halfdiamond*, mark size=2pt] table [x = {n}, y = {UnionBasic1}] {\pltA};
        \addlegendentry{\textsc{UnionBasic}}
        \addplot[blue, mark = Mercedes star, mark size=2pt] table [x = {n}, y = {NodeDeg0}] {\pltA};
        \addlegendentry{\textsc{NDegMin}}
        \addplot[purple, mark = pentagon*, mark size=2pt] table [x = {n}, y = {ForestsDFSExplicit1}] {\pltA};
        \addlegendentry{\textsc{ForestsDFS}}
        \addplot[green, mark = triangle*, mark size=2pt] table [x = {n}, y = {ForestsBFSExplicit1}] {\pltA};
        \addlegendentry{\textsc{ForestsBFS}}
        \addplot[red, mark = otimes*, mark size=1.75pt] table [x = {n}, y = {NodeInDegMin1}] {\pltA};
        \addlegendentry{\textsc{NInDegMin}}
        \addplot[brown, mark = 10-pointed star, mark size=2pt] table [x = {n}, y = {NodeProcMin0}] {\pltA};
        \addlegendentry{\textsc{NProcMin}}
        \addplot[brown, mark = halfcircle*, mark size=1.75pt] table [x = {n}, y = {PForestsDFS21}] {\pltA};
        \addlegendentry{\textsc{PForestsDFS}}
        \addplot[green, mark = oplus, mark size=1.75pt] table [x = {n}, y = {Basic1}] {\pltA};
        \addlegendentry{\textsc{Basic}}
        \addplot[orange, mark = halfcircle, mark size=1.75pt] table [x = {n}, y = {NodeBasic0}] {\pltA};
        \addlegendentry{\textsc{NBasic}}
        \addplot[blue, mark = diamond*, mark size=2pt] table [x = {n}, y = {PForestsBFS21}] {\pltA};
        \addlegendentry{\textsc{PForestsBFS}}
      \end{axis}
    \end{tikzpicture}
  }%
  \hfill
  \subfloat[Protein graphs] {
    \pgfplotstableread[col sep=comma]{./results/heuristics/proteins/Full/all.csv}{\pltA}
    \begin{tikzpicture}[scale=.95, trim axis left, trim axis right]
      \begin{axis}[
        xmin = 535.0, xmax = 10081.0,
        ymin = 0, ymax = 0.05,
        xtick distance = 2000.0,
        x tick label style={/pgf/number format/1000 sep = \kern 0.15em},
        tick label style={font=\scriptsize},
        label style={font=\scriptsize},
        grid = both,
        minor tick num = 4,
        major grid style = {lightgray},
        minor grid style = {lightgray!25},
        legend style={nodes={scale=0.5, transform shape},legend pos = outer south,legend columns=3,reverse legend},
        xtick={535.0, 2500.0, 5000.0, 7500.0, 10081.0},
        minor xtick={750.0, 1000.0, ..., 10000.00}
        ]
        \pgfplotsset{scaled x ticks=false}
        \addplot[purple, mark = Mercedes star flipped, mark size=2pt] table [x = {n}, y = {NodeDegComp0}] {\pltA};
        \addlegendentry{\textsc{NDegMinComp}}
        \addplot[orange, mark = square*, mark size=1.5pt] table [x = {n}, y = {NodeInDegMinComp0}] {\pltA};
        \addlegendentry{\textsc{NInDegMinComp}}
        \addplot[magenta, mark = oplus*, mark size=1.75pt] table [x = {n}, y = {NodeProcMinComp1}] {\pltA};
        \addlegendentry{\textsc{NProcMinComp}}
        \addplot[green, mark = heart, mark size=2pt] table [x = {n}, y = {NodeBasicComp1}] {\pltA};
        \addlegendentry{\textsc{NBasicComp}}
        \addplot[red, mark = halfsquare right*, mark size=2pt] table [x = {n}, y = {UnionNodeBasic1}] {\pltA};
        \addlegendentry{\textsc{UnionNBasic}}
        \addplot[purple, mark = square, mark size=1.5pt] table [x = {n}, y = {IncProcMin0}] {\pltA};
        \addlegendentry{\textsc{IncProcMin}}
        \addplot[orange, mark = halfsquare left*, mark size=2pt] table [x = {n}, y = {UnionTranspOneOut1}] {\pltA};
        \addlegendentry{\textsc{UnionTranspOne}}
        \addplot[magenta, mark = diamond, mark size=2pt] table [x = {n}, y = {Transp0}] {\pltA};
        \addlegendentry{\textsc{Transp}}
        \addplot[brown, mark = triangle, mark size=2pt] table [x = {n}, y = {IncInDegMin0}] {\pltA};
        \addlegendentry{\textsc{IncInDegMin}}
        \addplot[brown, mark = 10-pointed star, mark size=2pt] table [x = {n}, y = {NodeProcMin0}] {\pltA};
        \addlegendentry{\textsc{NProcMin}}
        \addplot[red, mark = otimes*, mark size=1.75pt] table [x = {n}, y = {NodeInDegMin1}] {\pltA};
        \addlegendentry{\textsc{NInDegMin}}
        \addplot[red, mark = pentagon, mark size=2pt] table [x = {n}, y = {TranspOneOut0}] {\pltA};
        \addlegendentry{\textsc{TranspOne}}
        \addplot[blue, mark = otimes, mark size=1.75pt] table [x = {n}, y = {DegMin0}] {\pltA};
        \addlegendentry{\textsc{DegMin}}
        \addplot[magenta, mark = halfdiamond*, mark size=2pt] table [x = {n}, y = {UnionBasic1}] {\pltA};
        \addlegendentry{\textsc{UnionBasic}}
        \addplot[brown, mark = halfcircle*, mark size=1.75pt] table [x = {n}, y = {PForestsDFS21}] {\pltA};
        \addlegendentry{\textsc{PForestsDFS}}
        \addplot[green, mark = triangle*, mark size=2pt] table [x = {n}, y = {ForestsBFSExplicit1}] {\pltA};
        \addlegendentry{\textsc{ForestsBFS}}
        \addplot[green, mark = oplus, mark size=1.75pt] table [x = {n}, y = {Basic1}] {\pltA};
        \addlegendentry{\textsc{Basic}}
        \addplot[purple, mark = pentagon*, mark size=2pt] table [x = {n}, y = {ForestsDFSExplicit1}] {\pltA};
        \addlegendentry{\textsc{ForestsDFS}}
        \addplot[blue, mark = Mercedes star, mark size=2pt] table [x = {n}, y = {NodeDeg0}] {\pltA};
        \addlegendentry{\textsc{NDegMin}}
        \addplot[orange, mark = halfcircle, mark size=1.75pt] table [x = {n}, y = {NodeBasic0}] {\pltA};
        \addlegendentry{\textsc{NBasic}}
        \addplot[blue, mark = diamond*, mark size=2pt] table [x = {n}, y = {PForestsBFS21}] {\pltA};
        \addlegendentry{\textsc{PForestsBFS}}
      \end{axis}
    \end{tikzpicture}
    \hspace{5mm}
  }
  \caption{Molecular and protein graphs.
    The \textsc{PForestsBFS} heuristic outperforms all others, highlighting the advantage of preprocessing steps that extract disjoint pseudoforests in real-world problem instances.}
\end{figure}
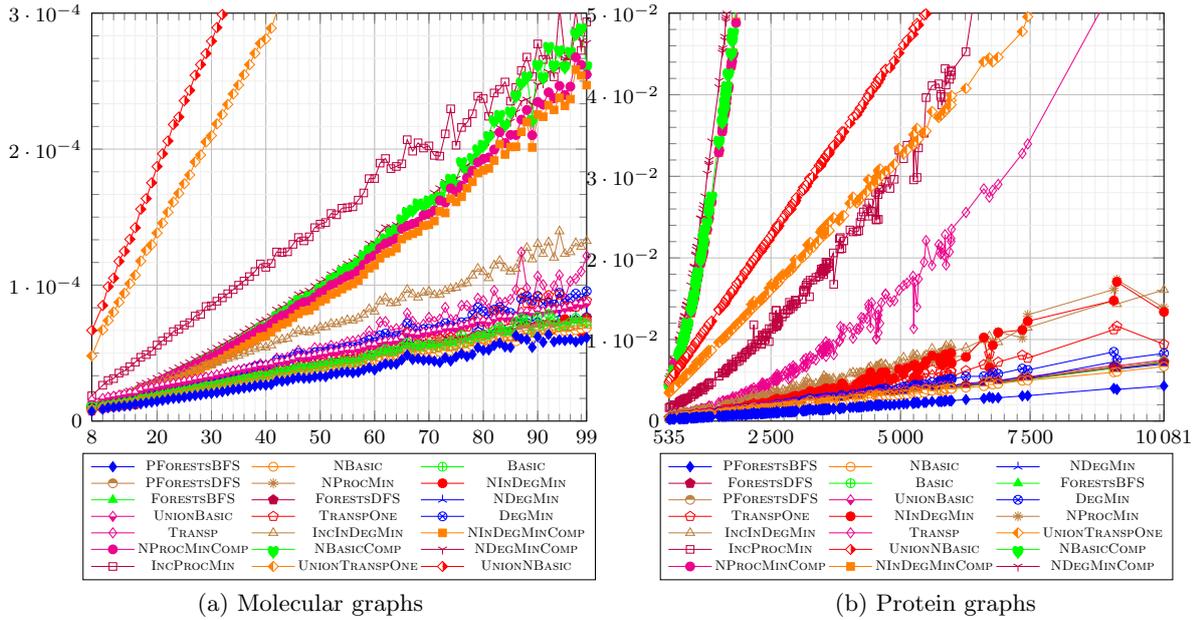

The benchmark results demonstrate that selecting edge-ordering heuristics based on graph properties significantly enhances performance in practice.
The \textsc{Transp} and \textsc{TranspOne} heuristics achieve the best results for Erd\H os--R\'enyi, Barab\'asi--Albert, and Tight graphs.
For dense random graphs, \textsc{NBasic}, \textsc{NProcMin}, and \textsc{NInDegMin} perform best.
The \textsc{PForestsBFS} heuristic is particularly effective for molecular and protein graphs.
Overall, \textsc{Transp} and \textsc{TranspOne} prove to be the most robust heuristics across diverse graph types, while \textsc{PForestsBFS} offers the best performance for molecular and protein graphs.

\medskip
Note that the component-based algorithm has a theoretical complexity of $O(n^2 + m)$, whereas the augmenting path method takes $O(nm)$ steps, making the component-based approach preferable for dense graphs.
In practice, however, the augmenting path method is more efficient when heuristics are used to optimize edge processing order.

\medskip
Next, we compare our implementation against two existing solutions: \textsc{KINARI}~\cite{kinariArticle, kinari} and \textsc{RAugm}~\cite{KiralyMihalykoGlobalrigidityAugmentation, rigidityAugm}.
Both competitors implement the component-based algorithm, just as our heuristic called \textsc{NInDegMinComp}.

\newpage
\begin{figure}[H]
  \setcounter{subfigure}{0}
  \centering
  \subfloat[$p=0.05$] {
    \hspace{5mm}
    \pgfplotstableread[col sep=comma]{./results/competitors/erdos_renyi/multiplicity_1/prob_0.05/Full/all.csv}{\pltA}
    \begin{tikzpicture}[scale=.95, trim axis left, trim axis right]
      \begin{axis}[
        xmin = 100.0, xmax = 1000.0,
        ymin = 0, ymax = 0.25,
        xtick distance = 200.0,
        x tick label style={/pgf/number format/1000 sep = \kern 0.15em},
        tick label style={font=\scriptsize},
        label style={font=\scriptsize},
        grid = both,
        minor tick num = 4,
        major grid style = {lightgray},
        minor grid style = {lightgray!25},
        legend style={nodes={scale=0.5, transform shape},legend pos = outer south,legend columns=3,reverse legend},
        xtick={100.0, 200.0, 300.0, 400.0, 500.0, 600.0, 700.0, 800.0, 900.0, 1000.0}
        ]
        \addplot[red, mark = pentagon, mark size=2pt] table [x = {n}, y = {Kinari}] {\pltA};
        \addlegendentry{\textsc{KINARI}}
        \addplot[brown, mark = triangle, mark size=2pt] table [x = {n}, y = {Augmentation1}] {\pltA};
        \addlegendentry{\textsc{RAugm}}
        \addplot[green, mark = oplus, mark size=1.75pt] table [x = {n}, y = {Basic1}] {\pltA};
        \addlegendentry{\textsc{Basic}}
        \addplot[blue, mark = diamond*, mark size=2pt] table [x = {n}, y = {PForestsBFS21}] {\pltA};
        \addlegendentry{\textsc{PForestsBFS}}
        \addplot[orange, mark = square*, mark size=1.5pt] table [x = {n}, y = {NodeInDegMinComp0}] {\pltA};
        \addlegendentry{\textsc{NInDegMinComp}}
        \addplot[magenta, mark = diamond, mark size=2pt] table [x = {n}, y = {Transp0}] {\pltA};
        \addlegendentry{\textsc{Transp}}
      \end{axis}
    \end{tikzpicture}
  }%
  \hfill
  \subfloat[$p=0.1$] {
    \pgfplotstableread[col sep=comma]{./results/competitors/erdos_renyi/multiplicity_1/prob_0.10/Full/all.csv}{\pltA}
    \begin{tikzpicture}[scale=.95, trim axis left, trim axis right]
      \begin{axis}[
        xmin = 100.0, xmax = 1000.0,
        ymin = 0, ymax = 0.25,
        xtick distance = 200.0,
        x tick label style={/pgf/number format/1000 sep = \kern 0.15em},
        tick label style={font=\scriptsize},
        label style={font=\scriptsize},
        grid = both,
        minor tick num = 4,
        major grid style = {lightgray},
        minor grid style = {lightgray!25},
        legend style={nodes={scale=0.5, transform shape},legend pos = outer south,legend columns=3,reverse legend},
        xtick={100.0, 200.0, 300.0, 400.0, 500.0, 600.0, 700.0, 800.0, 900.0, 1000.0}
        ]
        \addplot[red, mark = pentagon, mark size=2pt] table [x = {n}, y = {Kinari}] {\pltA};
        \addlegendentry{\textsc{KINARI}}
        \addplot[brown, mark = triangle, mark size=2pt] table [x = {n}, y = {Augmentation1}] {\pltA};
        \addlegendentry{\textsc{RAugm}}
        \addplot[green, mark = oplus, mark size=1.75pt] table [x = {n}, y = {Basic1}] {\pltA};
        \addlegendentry{\textsc{Basic}}
        \addplot[blue, mark = diamond*, mark size=2pt] table [x = {n}, y = {PForestsBFS21}] {\pltA};
        \addlegendentry{\textsc{PForestsBFS}}
        \addplot[orange, mark = square*, mark size=1.5pt] table [x = {n}, y = {NodeInDegMinComp0}] {\pltA};
        \addlegendentry{\textsc{NInDegMinComp}}
        \addplot[magenta, mark = diamond, mark size=2pt] table [x = {n}, y = {Transp0}] {\pltA};
        \addlegendentry{\textsc{Transp}}
      \end{axis}
    \end{tikzpicture}
    \hspace{5mm}
  }
  \newline
  \centering
  \subfloat[$p=0.5$] {
    \hspace{5mm}
    \pgfplotstableread[col sep=comma]{./results/competitors/erdos_renyi/multiplicity_1/prob_0.50/Full/all.csv}{\pltA}
    \begin{tikzpicture}[scale=.95, trim axis left, trim axis right]
      \begin{axis}[
        xmin = 100.0, xmax = 1000.0,
        ymin = 0, ymax = 0.5,
        xtick distance = 200.0,
        x tick label style={/pgf/number format/1000 sep = \kern 0.15em},
        tick label style={font=\scriptsize},
        label style={font=\scriptsize},
        grid = both,
        minor tick num = 4,
        major grid style = {lightgray},
        minor grid style = {lightgray!25},
        legend style={nodes={scale=0.5, transform shape},legend pos = outer south,legend columns=3,reverse legend},
        xtick={100.0, 200.0, 300.0, 400.0, 500.0, 600.0, 700.0, 800.0, 900.0, 1000.0}
        ]
        \addplot[red, mark = pentagon, mark size=2pt] table [x = {n}, y = {Kinari}] {\pltA};
        \addlegendentry{\textsc{KINARI}}
        \addplot[brown, mark = triangle, mark size=2pt] table [x = {n}, y = {Augmentation1}] {\pltA};
        \addlegendentry{\textsc{RAugm}}
        \addplot[green, mark = oplus, mark size=1.75pt] table [x = {n}, y = {Basic1}] {\pltA};
        \addlegendentry{\textsc{Basic}}
        \addplot[blue, mark = diamond*, mark size=2pt] table [x = {n}, y = {PForestsBFS21}] {\pltA};
        \addlegendentry{\textsc{PForestsBFS}}
        \addplot[orange, mark = square*, mark size=1.5pt] table [x = {n}, y = {NodeInDegMinComp0}] {\pltA};
        \addlegendentry{\textsc{NInDegMinComp}}
        \addplot[magenta, mark = diamond, mark size=2pt] table [x = {n}, y = {Transp0}] {\pltA};
        \addlegendentry{\textsc{Transp}}
      \end{axis}
    \end{tikzpicture}
  }%
  \hfill
  \subfloat[$p=0.9$] {
    \pgfplotstableread[col sep=comma]{./results/competitors/erdos_renyi/multiplicity_1/prob_0.90/Full/all.csv}{\pltA}
    \begin{tikzpicture}[scale=.95, trim axis left, trim axis right]
      \begin{axis}[
        xmin = 100.0, xmax = 1000.0,
        ymin = 0, ymax = 0.5,
        xtick distance = 200.0,
        x tick label style={/pgf/number format/1000 sep = \kern 0.15em},
        tick label style={font=\scriptsize},
        label style={font=\scriptsize},
        grid = both,
        minor tick num = 4,
        major grid style = {lightgray},
        minor grid style = {lightgray!25},
        legend style={nodes={scale=0.5, transform shape},legend pos = outer south,legend columns=3,reverse legend},
        xtick={100.0, 200.0, 300.0, 400.0, 500.0, 600.0, 700.0, 800.0, 900.0, 1000.0}
        ]
        \addplot[red, mark = pentagon, mark size=2pt] table [x = {n}, y = {Kinari}] {\pltA};
        \addlegendentry{\textsc{KINARI}}
        \addplot[brown, mark = triangle, mark size=2pt] table [x = {n}, y = {Augmentation1}] {\pltA};
        \addlegendentry{\textsc{RAugm}}
        \addplot[green, mark = oplus, mark size=1.75pt] table [x = {n}, y = {Basic1}] {\pltA};
        \addlegendentry{\textsc{Basic}}
        \addplot[blue, mark = diamond*, mark size=2pt] table [x = {n}, y = {PForestsBFS21}] {\pltA};
        \addlegendentry{\textsc{PForestsBFS}}
        \addplot[orange, mark = square*, mark size=1.5pt] table [x = {n}, y = {NodeInDegMinComp0}] {\pltA};
        \addlegendentry{\textsc{NInDegMinComp}}
        \addplot[magenta, mark = diamond, mark size=2pt] table [x = {n}, y = {Transp0}] {\pltA};
        \addlegendentry{\textsc{Transp}}
      \end{axis}
    \end{tikzpicture}
    \hspace{5mm}
  }
  \caption{Erd\H os--R\'enyi random graphs.
    Although \textsc{NInDegMinComp}, \textsc{RAugm}, and \textsc{KINARI} implement the same component-based approach, their efficiency varies significantly.
    \textsc{NInDegMinComp} is the fastest among them, but \textsc{Transp} outperforms it, particularly for sparse graphs.}
\end{figure}
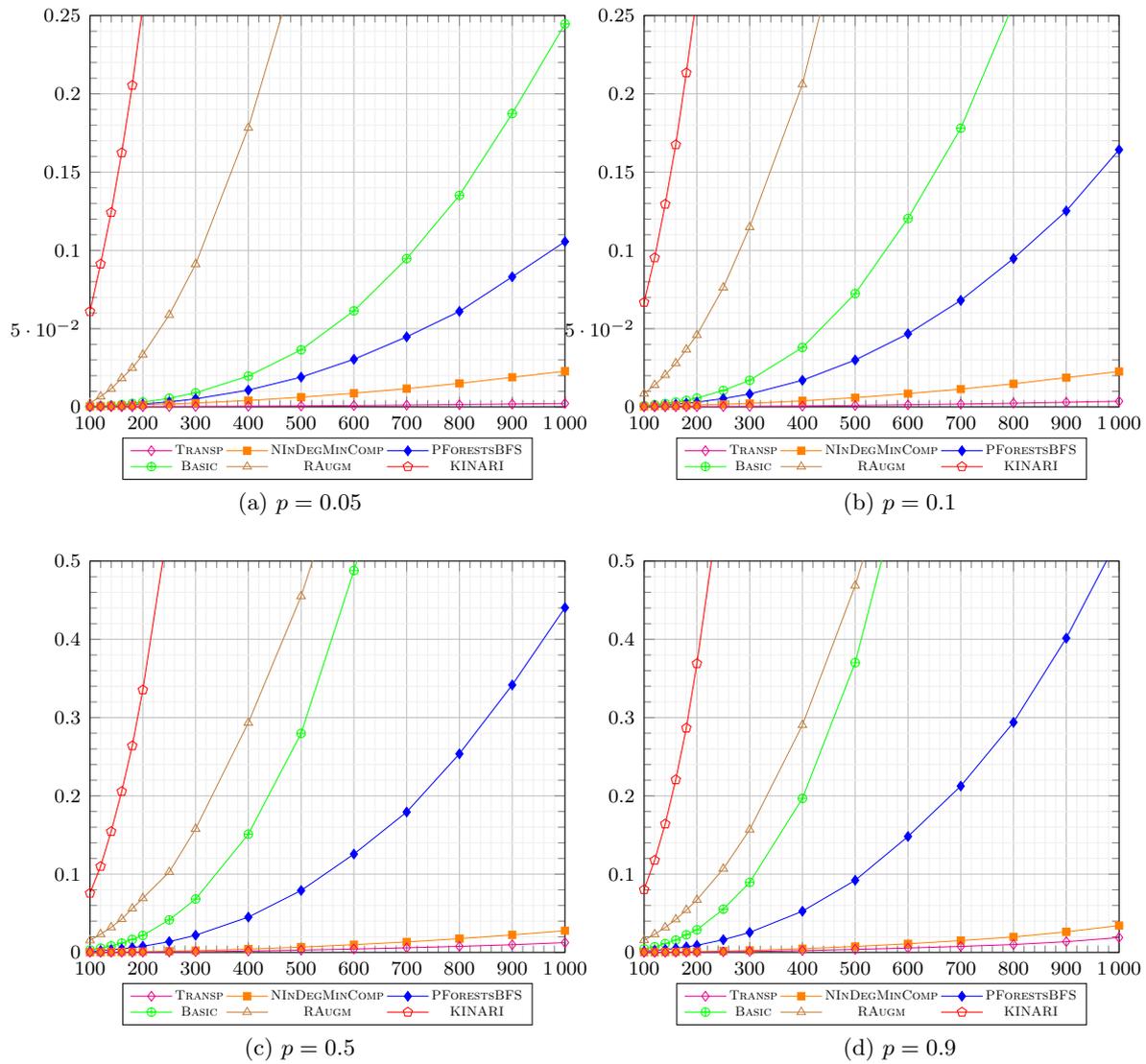

\begin{figure}[H]
  \setcounter{subfigure}{0}
  \centering
  \subfloat[$m=5$] {
    \hspace{5mm}
    \pgfplotstableread[col sep=comma]{./results/competitors/barabasi_albert/edges_5/Full/all.csv}{\pltA}
    \begin{tikzpicture}[scale=.95, trim axis left, trim axis right]
      \begin{axis}[
        xmin = 100.0, xmax = 1000.0,
        ymin = 0, ymax = 0.025,
        xtick distance = 200.0,
        x tick label style={/pgf/number format/1000 sep = \kern 0.15em},
        tick label style={font=\scriptsize},
        label style={font=\scriptsize},
        grid = both,
        minor tick num = 4,
        major grid style = {lightgray},
        minor grid style = {lightgray!25},
        legend style={nodes={scale=0.5, transform shape},legend pos = outer south,legend columns=3,reverse legend},
        xtick={100.0, 200.0, 300.0, 400.0, 500.0, 600.0, 700.0, 800.0, 900.0, 1000.0}
        ]
        \addplot[red, mark = pentagon, mark size=2pt] table [x = {n}, y = {Kinari}] {\pltA};
        \addlegendentry{\textsc{KINARI}}
        \addplot[brown, mark = triangle, mark size=2pt] table [x = {n}, y = {Augmentation1}] {\pltA};
        \addlegendentry{\textsc{RAugm}}
        \addplot[orange, mark = square*, mark size=1.5pt] table [x = {n}, y = {NodeInDegMinComp0}] {\pltA};
        \addlegendentry{\textsc{NInDegMinComp}}
        \addplot[green, mark = oplus, mark size=1.75pt] table [x = {n}, y = {Basic1}] {\pltA};
        \addlegendentry{\textsc{Basic}}
        \addplot[blue, mark = diamond*, mark size=2pt] table [x = {n}, y = {PForestsBFS21}] {\pltA};
        \addlegendentry{\textsc{PForestsBFS}}
        \addplot[magenta, mark = diamond, mark size=2pt] table [x = {n}, y = {Transp0}] {\pltA};
        \addlegendentry{\textsc{Transp}}
      \end{axis}
    \end{tikzpicture}
  }%
  \hfill
  \subfloat[$m=50$] {
    \pgfplotstableread[col sep=comma]{./results/competitors/barabasi_albert/edges_50/Full/all.csv}{\pltA}
    \begin{tikzpicture}[scale=.95, trim axis left, trim axis right]
      \begin{axis}[
        xmin = 100.0, xmax = 1000.0,
        ymin = 0, ymax = 0.1,
        xtick distance = 200.0,
        x tick label style={/pgf/number format/1000 sep = \kern 0.15em},
        tick label style={font=\scriptsize},
        label style={font=\scriptsize},
        grid = both,
        minor tick num = 4,
        major grid style = {lightgray},
        minor grid style = {lightgray!25},
        legend style={nodes={scale=0.5, transform shape},legend pos = outer south,legend columns=3,reverse legend},
        xtick={100.0, 200.0, 300.0, 400.0, 500.0, 600.0, 700.0, 800.0, 900.0, 1000.0}
        ]
        \addplot[red, mark = pentagon, mark size=2pt] table [x = {n}, y = {Kinari}] {\pltA};
        \addlegendentry{\textsc{KINARI}}
        \addplot[brown, mark = triangle, mark size=2pt] table [x = {n}, y = {Augmentation1}] {\pltA};
        \addlegendentry{\textsc{RAugm}}
        \addplot[green, mark = oplus, mark size=1.75pt] table [x = {n}, y = {Basic1}] {\pltA};
        \addlegendentry{\textsc{Basic}}
        \addplot[blue, mark = diamond*, mark size=2pt] table [x = {n}, y = {PForestsBFS21}] {\pltA};
        \addlegendentry{\textsc{PForestsBFS}}
        \addplot[orange, mark = square*, mark size=1.5pt] table [x = {n}, y = {NodeInDegMinComp0}] {\pltA};
        \addlegendentry{\textsc{NInDegMinComp}}
        \addplot[magenta, mark = diamond, mark size=2pt] table [x = {n}, y = {Transp0}] {\pltA};
        \addlegendentry{\textsc{Transp}}
      \end{axis}
    \end{tikzpicture}
    \hspace{5mm}
  }
  \caption{Barab\'asi--Albert random graphs.
    As with Erd\H os--R\'enyi graphs, \textsc{Transp} is the most efficient.
    The running time of \textsc{KINARI} was $6 \cdot 10^{-2}$ seconds for 100 nodes with $m=5$, so it is omitted from the figure.}
\end{figure}
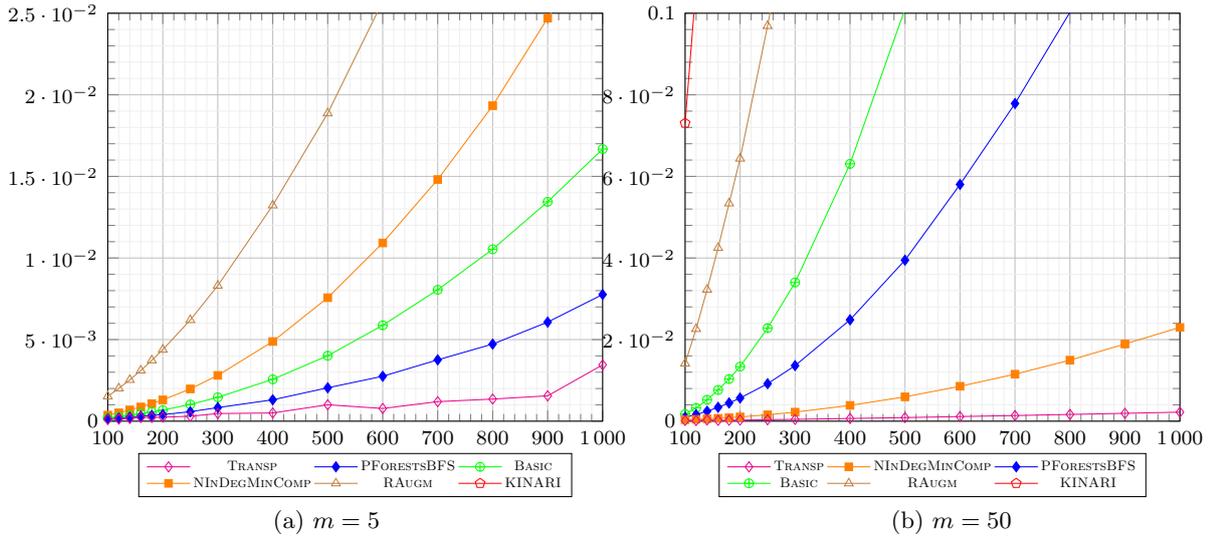

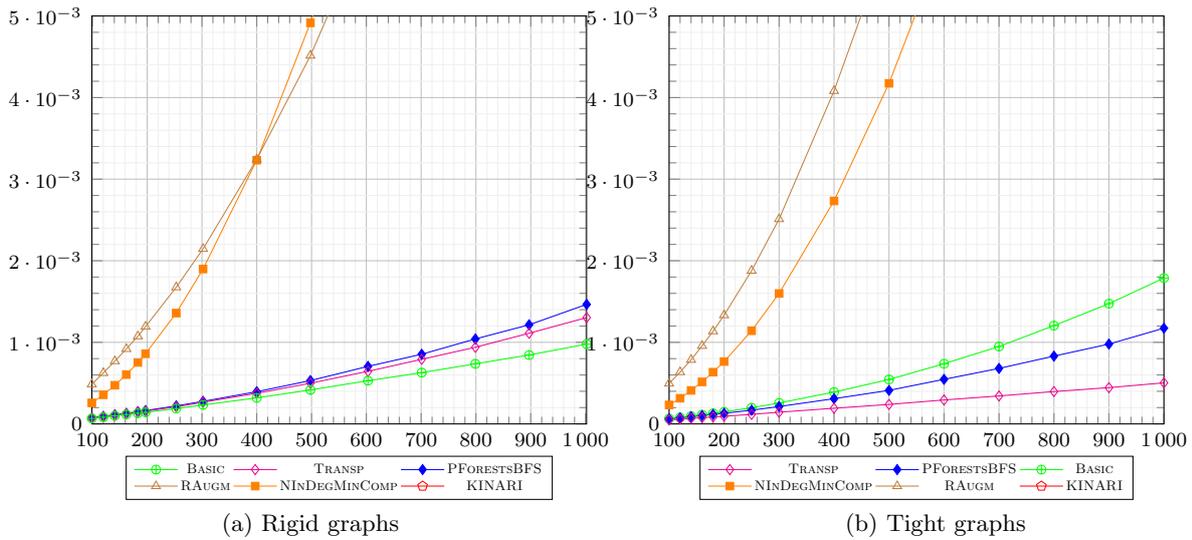
\begin{figure}[H]
  \setcounter{subfigure}{0}
  \centering
  \subfloat[Rigid graphs] {
    \hspace{5mm}
    \pgfplotstableread[col sep=comma]{./results/competitors/rigids/Full/all.csv}{\pltA}
    \begin{tikzpicture}[scale=.95, trim axis left, trim axis right]
      \begin{axis}[
        xmin = 99.0, xmax = 1002.0,
        ymin = 0, ymax = 0.005,
        xtick distance = 200.0,
        x tick label style={/pgf/number format/1000 sep = \kern 0.15em},
        tick label style={font=\scriptsize},
        label style={font=\scriptsize},
        grid = both,
        minor tick num = 4,
        major grid style = {lightgray},
        minor grid style = {lightgray!25},
        legend style={nodes={scale=0.5, transform shape},legend pos = outer south,legend columns=3,reverse legend},
        xtick={100.0, 200.0, 300.0, 400.0, 500.0, 600.0, 700.0, 800.0, 900.0, 1000.0}
        ]
        \addplot[red, mark = pentagon, mark size=2pt] table [x = {n}, y = {Kinari}] {\pltA};
        \addlegendentry{\textsc{KINARI}}
        \addplot[orange, mark = square*, mark size=1.5pt] table [x = {n}, y = {NodeInDegMinComp0}] {\pltA};
        \addlegendentry{\textsc{NInDegMinComp}}
        \addplot[brown, mark = triangle, mark size=2pt] table [x = {n}, y = {Augmentation1}] {\pltA};
        \addlegendentry{\textsc{RAugm}}
        \addplot[blue, mark = diamond*, mark size=2pt] table [x = {n}, y = {PForestsBFS21}] {\pltA};
        \addlegendentry{\textsc{PForestsBFS}}
        \addplot[magenta, mark = diamond, mark size=2pt] table [x = {n}, y = {Transp0}] {\pltA};
        \addlegendentry{\textsc{Transp}}
        \addplot[green, mark = oplus, mark size=1.75pt] table [x = {n}, y = {Basic1}] {\pltA};
        \addlegendentry{\textsc{Basic}}
      \end{axis}
    \end{tikzpicture}
  }%
  \hfill
  \subfloat[Tight graphs] {
    \pgfplotstableread[col sep=comma]{./results/competitors/tights/Full/all.csv}{\pltA}
    \begin{tikzpicture}[scale=.95, trim axis left, trim axis right]
      \begin{axis}[
        xmin = 100.0, xmax = 1000.0,
        ymin = 0, ymax = 0.005,
        xtick distance = 200.0,
        x tick label style={/pgf/number format/1000 sep = \kern 0.15em},
        tick label style={font=\scriptsize},
        label style={font=\scriptsize},
        grid = both,
        minor tick num = 4,
        major grid style = {lightgray},
        minor grid style = {lightgray!25},
        legend style={nodes={scale=0.5, transform shape},legend pos = outer south,legend columns=3,reverse legend},
        xtick={100.0, 200.0, 300.0, 400.0, 500.0, 600.0, 700.0, 800.0, 900.0, 1000.0}
        ]
        \addplot[red, mark = pentagon, mark size=2pt] table [x = {n}, y = {Kinari}] {\pltA};
        \addlegendentry{\textsc{KINARI}}
        \addplot[brown, mark = triangle, mark size=2pt] table [x = {n}, y = {Augmentation1}] {\pltA};
        \addlegendentry{\textsc{RAugm}}
        \addplot[orange, mark = square*, mark size=1.5pt] table [x = {n}, y = {NodeInDegMinComp0}] {\pltA};
        \addlegendentry{\textsc{NInDegMinComp}}
        \addplot[green, mark = oplus, mark size=1.75pt] table [x = {n}, y = {Basic1}] {\pltA};
        \addlegendentry{\textsc{Basic}}
        \addplot[blue, mark = diamond*, mark size=2pt] table [x = {n}, y = {PForestsBFS21}] {\pltA};
        \addlegendentry{\textsc{PForestsBFS}}
        \addplot[magenta, mark = diamond, mark size=2pt] table [x = {n}, y = {Transp0}] {\pltA};
        \addlegendentry{\textsc{Transp}}
      \end{axis}
    \end{tikzpicture}
    \hspace{5mm}
  }
  \caption{Rigid and tight graphs.
    For rigid graphs, \textsc{RAugm} outperforms \textsc{NInDegMinComp} for large instances.
    While \textsc{Basic} could not be improved for rigid graphs, \textsc{Transp} performs significantly better for dense $(k, k)$-tight graphs.
    The running time of \textsc{KINARI} was $3 \cdot 10^{-2}$ seconds for 100 nodes in both cases, so it is omitted from the figures.}
\end{figure}

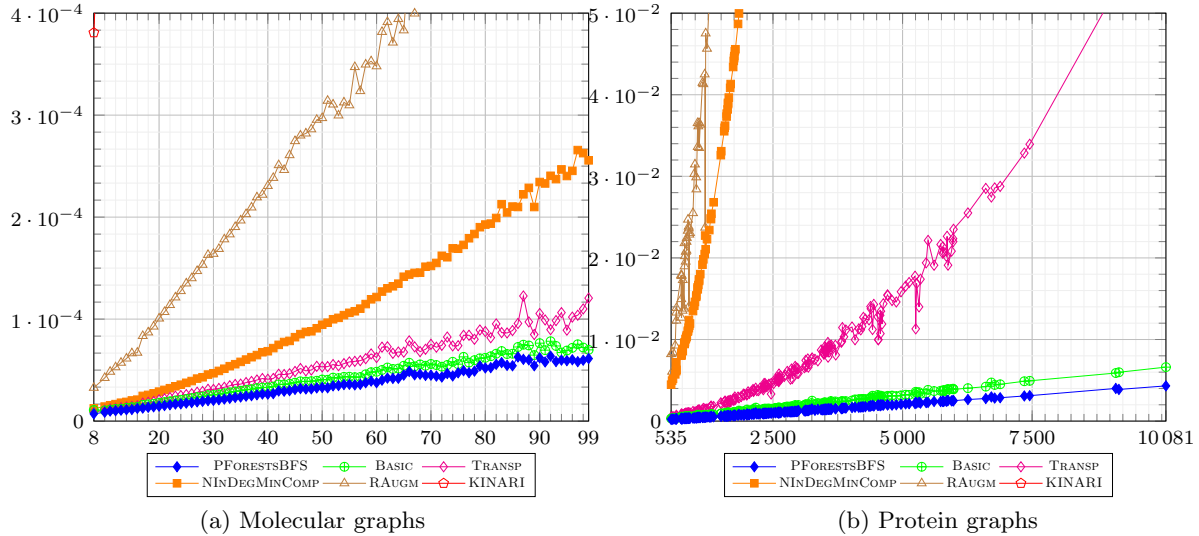
\begin{figure}[H]
  \setcounter{subfigure}{0}
  \centering
  \subfloat[Molecular graphs] {
    \hspace{5mm}
    \pgfplotstableread[col sep=comma]{./results/competitors/moleculars/Full/all.csv}{\pltA}
    \begin{tikzpicture}[scale=.95, trim axis left, trim axis right]
      \begin{axis}[
        xmin = 8.0, xmax = 99.0,
        ymin = 0, ymax = 0.0004,
        xtick distance = 10.0,
        x tick label style={/pgf/number format/1000 sep = \kern 0.15em},
        tick label style={font=\scriptsize},
        label style={font=\scriptsize},
        grid = both,
        minor tick num = 5,
        major grid style = {lightgray},
        minor grid style = {lightgray!25},
        legend style={nodes={scale=0.5, transform shape},legend pos = outer south,legend columns=3,reverse legend},
        xtick={8.0, 20.0, 30.0, 40.0, 50.0, 60.0, 70.0, 80.0, 90.0, 99.0},
        minor xtick={8, 10, ..., 98}
        ]
        \addplot[red, mark = pentagon, mark size=2pt] table [x = {n}, y = {Kinari}] {\pltA};
        \addlegendentry{\textsc{KINARI}}
        \addplot[brown, mark = triangle, mark size=2pt] table [x = {n}, y = {Augmentation1}] {\pltA};
        \addlegendentry{\textsc{RAugm}}
        \addplot[orange, mark = square*, mark size=1.5pt] table [x = {n}, y = {NodeInDegMinComp0}] {\pltA};
        \addlegendentry{\textsc{NInDegMinComp}}
        \addplot[magenta, mark = diamond, mark size=2pt] table [x = {n}, y = {Transp0}] {\pltA};
        \addlegendentry{\textsc{Transp}}
        \addplot[green, mark = oplus, mark size=1.75pt] table [x = {n}, y = {Basic1}] {\pltA};
        \addlegendentry{\textsc{Basic}}
        \addplot[blue, mark = diamond*, mark size=2pt] table [x = {n}, y = {PForestsBFS21}] {\pltA};
        \addlegendentry{\textsc{PForestsBFS}}
      \end{axis}
    \end{tikzpicture}
  }%
  \hfill
  \subfloat[Protein graphs] {
    \pgfplotstableread[col sep=comma]{./results/competitors/proteins/Full/all.csv}{\pltA}
    \begin{tikzpicture}[scale=.95, trim axis left, trim axis right]
      \begin{axis}[
        xmin = 535.0, xmax = 10081.0,
        ymin = 0, ymax = 0.05,
        xtick distance = 2000.0,
        x tick label style={/pgf/number format/1000 sep = \kern 0.15em},
        tick label style={font=\scriptsize},
        label style={font=\scriptsize},
        grid = both,
        minor tick num = 4,
        major grid style = {lightgray},
        minor grid style = {lightgray!25},
        legend style={nodes={scale=0.5, transform shape},legend pos = outer south,legend columns=3,reverse legend},
        xtick={535.0, 2500.0, 5000.0, 7500.0, 10081.0},
        minor xtick={750.0, 1000.0, ..., 10000.00}
        ]
        \pgfplotsset{scaled x ticks=false}
        \addplot[red, mark = pentagon, mark size=2pt] table [x = {n}, y = {Kinari}] {\pltA};
        \addlegendentry{\textsc{KINARI}}
        \addplot[brown, mark = triangle, mark size=2pt] table [x = {n}, y = {Augmentation1}] {\pltA};
        \addlegendentry{\textsc{RAugm}}
        \addplot[orange, mark = square*, mark size=1.5pt] table [x = {n}, y = {NodeInDegMinComp0}] {\pltA};
        \addlegendentry{\textsc{NInDegMinComp}}
        \addplot[magenta, mark = diamond, mark size=2pt] table [x = {n}, y = {Transp0}] {\pltA};
        \addlegendentry{\textsc{Transp}}
        \addplot[green, mark = oplus, mark size=1.75pt] table [x = {n}, y = {Basic1}] {\pltA};
        \addlegendentry{\textsc{Basic}}
        \addplot[blue, mark = diamond*, mark size=2pt] table [x = {n}, y = {PForestsBFS21}] {\pltA};
        \addlegendentry{\textsc{PForestsBFS}}
      \end{axis}
    \end{tikzpicture}
    \hspace{5mm}
  }
  \caption{Molecular and protein graphs.
    The \textsc{PForestsBFS} heuristic is the most efficient.
    The running time of \textsc{KINARI} was $0.6$ seconds for protein graphs with 535 nodes, so it is omitted from the figure.}
\end{figure}

\section{Conclusion}\label{sec:conclusion}

We presented an efficient implementation for extracting a maximum-size $(k, \ell)$-sparse subgraph, building on the classical augmenting path method and enhancing it with a range of heuristics to optimize the processing order of edges.
For the special case $\ell = 2k$, we proposed an improved algorithm that computes an inclusion-wise maximal $(k, 2k)$-sparse subgraph in $O(nm)$ time.

Our experimental results show that the implementation consistently outperforms existing tools, often by several orders of magnitude, across a variety of graph families.
The evaluation highlights the impact of edge-ordering heuristics, with \textsc{Transp} and \textsc{TranspOne} performing best on synthetic instances, while \textsc{PForestsBFS} proves especially effective on real-world molecular and protein graphs.
These results clearly demonstrate the substantial practical benefits of preprocessing and ordering strategies.

The implementation is publicly available and proposed for inclusion in the LEMON graph library, offering a robust, scalable, and easy-to-use tool for applications in rigidity theory, combinatorial optimization, structural biology, robotics, and CAD.

\medskip
While our heuristics demonstrate strong empirical performance across a wide range of graph classes, a formal analysis of their worst-case behavior remains an open challenge --- though their theoretical running time is no worse than that of the original augmenting path method.
Future work could also explore heuristic strategies for the weighted case, where flexibility in edge-processing order is limited but freedom in choosing the orientation of accepted edges is retained.

\section*{Acknowledgements}

The author thanks L\'or\'ant Mat\'uz for his contributions to the early experimental setup and for his assistance with benchmarking and evaluations during the initial stages of this work.
The author is also grateful to Tibor Jord\'an, Csaba Kir\'aly, and Andr\'as Mih\'alyk\'o for valuable discussions and for pointing out relevant literature.

This research has been implemented with the support provided by the Ministry of Innovation and Technology of Hungary from the National Research, Development and Innovation Fund, financed under the ELTE TKP 2021-NKTA-62 funding scheme, by the Ministry of Innovation and Technology NRDI Office within the framework of the Artificial Intelligence National Laboratory Program, by the Lend\"ulet Programme of the Hungarian Academy of Sciences --- grant number LP2021-1/2021.

\bibliographystyle{plain}
\bibliography{bibliography}

\end{document}